\documentclass[11pt]{article}
\usepackage[utf8]{inputenc}
\usepackage{amsfonts,latexsym,amsthm,amssymb,amsmath,amscd,euscript,mathrsfs,graphicx}
\usepackage{framed}
\usepackage{fullpage}
\usepackage{color}
\usepackage{enumitem}
\usepackage[obeyFinal,textsize=scriptsize,shadow]{todonotes}
\usepackage{hyperref}
\usepackage{float}

\newtheorem{theorem}{Theorem}[section]
\newtheorem{proposition}[theorem]{Proposition}
\newtheorem{lemma}[theorem]{Lemma}
\theoremstyle{definition}
\newtheorem{definition}[theorem]{Definition}

\theoremstyle{remark}

\newcommand{\BC}{\mathbb C}
\newcommand{\BE}{\mathbb E}

\newcommand{\BN}{\mathbb N}

\newcommand{\BQ}{\mathbb Q}

\newcommand{\BZ}{\mathbb Z}

\newcommand{\poly}{\text{poly}}

\newcommand{\Mod}[1]{\ (\mathrm{mod}\ #1)}

\title{Circular Trace Reconstruction}
\author{Shyam Narayanan\thanks{Massachusetts Institute of Technology. \texttt{shyamsn@mit.edu}. Supported by the MIT Akamai Fellowship, the NSF Graduate Fellowship, and a Simons Investigator Award.}
\and
Michael Ren\thanks{Massachusetts Institute of Technology. \texttt{mren@mit.edu}. Supported by NSF-DMS grant 1949884 and NSA grant H98230-20-1-0009.}}
\date{\today}

\begin{document}

\maketitle

\begin{abstract}
    \emph{Trace reconstruction} is the problem of learning an unknown string $x$ from independent traces of $x$, where traces are generated by independently deleting each bit of $x$ with some deletion probability $q$. In this paper, we initiate the study of \emph{Circular trace reconstruction}, where the unknown string $x$ is circular and traces are now rotated by a random cyclic shift. Trace reconstruction is related to many computational biology problems studying DNA, which is a primary motivation for this problem as well, as many types of DNA are known to be circular.
    
    Our main results are as follows. First, we prove that we can reconstruct arbitrary circular strings of length $n$ using $\exp\big(\tilde{O}(n^{1/3})\big)$ traces for any constant deletion probability $q$, as long as $n$ is prime or the product of two primes. For $n$ of this form, this nearly matches what was the best known bound of $\exp\big(O(n^{1/3})\big)$ for standard trace reconstruction when this paper was initially released. We note, however, that Chase very recently improved the standard trace reconstruction bound to $\exp\big(\tilde{O}(n^{1/5})\big)$. Next, we prove that we can reconstruct random circular strings with high probability using $n^{O(1)}$ traces for any constant deletion probability $q$. Finally, we prove a lower bound of $\tilde{\Omega}(n^3)$ traces for arbitrary circular strings, which is greater than the best known lower bound of $\tilde{\Omega}(n^{3/2})$ in standard trace reconstruction. 
    %We conclude by noting some further research directions and open problems.
\end{abstract}
\newpage

\section{Introduction}

The trace reconstruction problem asks one to recover an unknown string $x$ of length $n$ from independent noisy samples of the string. In the original setting, $x$ is a binary string in $\{0, 1\}^n,$ and a random subsequence $\tilde{x}$ of $x$, called a \emph{trace}, is generated by sending $x$ through a deletion channel with deletion probability $q$, which removes each bit of $x$ independently with some fixed probability $q$. The main question is to determine how many independent traces are needed to recover the original string with high probability. This question has become very well studied over the past two decades \cite{Levenshtein01a, Levenshtein01b, BatuKKM04, KannanM05, HolensteinMPW08, ViswanathanS08, McGregorPV14, DeOS17, NazarovP17, PeresZ17, HartungHP18, HoldenL18, HoldenPP18, Chase19, ChenDLSS20, Chase20}, with many results over various settings. For instance, people have studied the case where we wish to reconstruct $x$ for any arbitrarily chosen $x \in \{0, 1\}^n$ (worst-case) or the case where we just wish to reconstruct a randomly chosen string $x$ (average-case). People have also studied the trace reconstruction problem for various values of the deletion probability $q$, such as if $q$ is a fixed constant between $0$ and $1$ or decays as some function of $n$. People have also studied variants where the traces allow for insertions of random bits, rather than just deletions, and variants where the string is no longer binary but from a larger alphabet.

Finally, various generalizations or variants of the trace reconstruction problem have also been developed. These include error-correcting codes over the deletion channel (i.e., ``coded'' trace reconstruction) \cite{CheraghchiGMR19, BrakensiekLS19}, reconstructing matrices \cite{KrishnamurthyMMP19} and trees \cite{DaviesRR19} from traces, and reconstructing mixtures of strings from traces \cite{BanCFSS19, BanCSS19, Narayanan20}.

In this paper, we develop and study a new variant of trace reconstruction that we call \emph{Circular trace reconstruction}. In this variant, there is again an unknown string $x \in \{0, 1\}^n$ that we can sample traces from, but this time, the string $x$ is a cyclic string, meaning that there is no beginning or end to the string. Equivalently, one can imagine a linear string that undergoes a random cyclic shift before a trace is returned. See Figure \ref{Example} for an example. Our goal, like in the normal trace reconstruction, is to reconstruct the original circular string using as few random traces as possible.

\begin{figure}
\centering
\begin{tikzpicture}
%\tikzset{bead/.style={circle, fill=white,draw=black,inner sep=1pt}}
\tikzstyle{bead}=[shape=circle, draw, inner sep=1,fill=white]

\draw (0, 0)
-- ++(75:0.6) node[bead] {\color{red} 1}
-- ++(45:0.6) node[bead] {\color{red} 0}
-- ++(15:0.6) node[bead] {\color{red} 0}
-- ++(-15:0.6) node[bead] {0}
-- ++(-45:0.6) node[bead] {\color{red} 1}
-- ++(-75:0.6) node[bead] {0}
-- ++(-105:0.6) node[bead] {1}
-- ++(-135:0.6) node[bead] {\color{red} 1}
-- ++(-165:0.6) node[bead] {0}
-- ++(-195:0.6) node[bead] {1}
-- ++(-225:0.6) node[bead] {\color{red}1}
-- ++(-255:0.6) node[bead] {0};

\draw[->, thick] (2.6, 0) -- (5.7, 0);

\draw (6,0)
-- ++(60:1) node[bead] {0}
-- ++(0:1) node[bead] {0}
-- ++(-60:1) node[bead] {1}
-- ++(-120:1) node[bead] {0}
-- ++(-180:1) node[bead] {1}
-- ++(-240:1) node[bead] {0};

\draw[->, thick] (8.3, 0) -- (11.2, 0);

\node[draw,text width=1.1cm] at (12,0) {010001};
\end{tikzpicture}
\caption{An example of a circular trace. We start with an unknown circular string (top left). Each bit of the string is randomly deleted (red bits are deleted, black bits are retained) and the order of the retained bits is preserved, so we are left with the smaller circular string. However, since there is no beginning or end of the circular string, we assume the string is seen in clockwise order starting from a randomly chosen bit.}
\label{Example}
\end{figure}
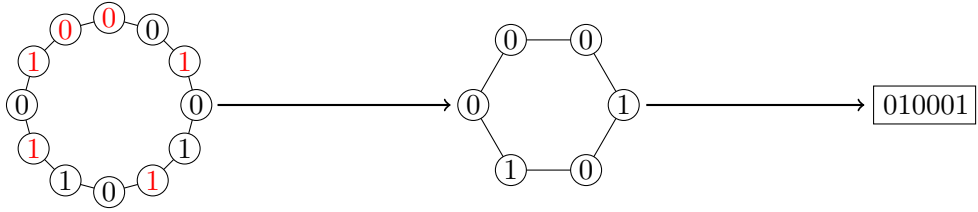

\subsection{Main Results and Comparison to Linear Trace Reconstruction}

Perhaps the first natural question about circular trace reconstruction is the following: how does the sample complexity of circular trace reconstruction compare to the sample complexity of standard (linear) trace reconstruction? Intuitively, one should expect circular trace reconstruction to be at least as difficult as standard trace reconstruction, since given any trace of a linear string, we can randomly rotate it to get a trace of the corresponding circular string. 
%This would seemingly imply that circular trace reconstruction is at least as hard as linear trace reconstruction. 
This reasoning, however, is slightly flawed. For instance, perhaps the hardest instance of linear trace reconstruction comes from distinguishing between two strings $x$ and $y$ which are different as linear strings but equivalent up to a cyclic shift. In this case, the circular trace reconstruction problem does not even need to distinguish between $x$ and $y$, because they are equivalent! However, by padding the trace with extra bits before randomly rotating, one can show that circular trace reconstruction is at least as hard as linear trace reconstruction in both the worst-case and average-case. Indeed, we have the following proposition -- as its proof is quite simple, we defer it to Appendix \ref{Omitted}.

\begin{proposition} \label{CircularHarderThanLinear}
    Suppose that we can solve worst-case (resp., average case) circular trace reconstruction over length $m$ strings with deletion probability $q$ using $T(m, q)$ traces. Then, we can solve worst-case (resp., average case) linear trace reconstruction over length $n$ strings with deletion probability $q$ using $\min_{m \ge 2n} T(m, q)$ traces.
    %
    %Likewise, suppose that we can solve average-case circular trace reconstruction over length $m$ strings with deletion probability $q$ using $T_2(m, q)$ traces. Then, we can solve average-case linear trace reconstruction over length $n$ strings with deletion probability $q$ using $\min_{m \ge 2n} T_2(m, q)$ traces.
\end{proposition}

Given Proposition \ref{CircularHarderThanLinear}, any upper bounds for circular trace reconstruction imply nearly equivalent upper bounds for the linear trace reconstruction, and any lower bounds for linear trace reconstruction imply nearly equivalent lower bounds for circular trace reconstruction. This raises two natural questions. First, can we match or nearly match the best linear trace reconstruction upper bounds for circular trace reconstruction? Second, can we beat the best linear trace reconstruction lower bounds for circular trace reconstruction?

\medskip

%Our main results can be split into three classes. First, we provide upper bounds for worst-case strings. Next, we provide upper bounds for average-case strings. Finally, we provide lower bounds for worst-case strings. For all of our results, we focus on the case where the deletion probabilities $q$ are fixed constants between $0$ and $1$, although our upper bound results also provide interesting bounds for regimes where $q = 1-o(1)$.

The first main result we prove is for worst-case circular strings. At the time of the initial release of this paper, the best known upper bound for worst-case linear trace reconstruction with deletion probability $q$, where $q$ is a fixed constant between $0$ and $1$, is $\exp\left(O(n^{1/3})\right),$ where the unknown string has length $n$ \cite{DeOS17, NazarovP17}. Shortly afterwards, the upper bound was improved to $\exp\left(\tilde{O}(n^{1/5})\right)$ \cite{Chase20}. Our first main result, which we prove in Section \ref{WorstCase}, provides an upper bound for the circular trace reconstruction problem that nearly matches the results of \cite{DeOS17, NazarovP17}, but only if the length $n$ has at most $2$ prime factors.

\begin{theorem} \label{WorstCaseThm}
    Let $x$ be an unknown, arbitrary circular string of length $n$, let $q$ be the deletion probability of each element in the string, and let $p = 1-q$ be the retention probability. Then, if $n$ is either a prime or a product of two (possibly equal) primes, using $\exp\left(O\left(n^{1/3} (\log n)^{2/3} p^{-2/3}\right)\right)$ random traces, we can determine $x$ with failure probability at most $2^{-n}.$
\end{theorem}

The primary reason why our theorem fails for $n$ having $3$ or more prime factors is that we prove the following number theoretic result which is crucial in our algorithm.

\begin{theorem} \label{NT:Main}
    For any fixed integer $n \ge 2,$ the following statement is true \textbf{if and only if} $n$ has at most $2$ prime factors, counting multiplicity. 
    
    Define $\omega := e^{2 \pi i/n}$, and suppose that $a_0, \dots, a_{n-1}, b_0, \dots, b_{n-1}$ are all integers in $\{0, 1\}.$ Also, suppose that for all $0 \le k \le n-1,$ there is some integer $c_k$ such that $\sum_{i = 1}^{n} a_i \omega^{i \cdot k} = \omega^{c_k} \cdot \sum_{i = 1}^{n} b_i \omega^{i \cdot k}.$ Then, the sequences $\{a_i\}$ and $\{b_i\}$ are cyclic shifts of each other.
\end{theorem}

%\begin{theorem} \label{Main2}
    %Let $x$ be an unknown circular string of length $n$, let $q$ be the deletion probability of each element, and let $p = 1-q$ be the retention probability of each element in the string. Then, regardless of $n$'s prime factorization, using $\exp\left(O\left(n^{1/2} (\log n)^{1/2} p^{-1/2}\right)\right)$ random traces, we can determine $x$ with failure probability at most $2^{-n}.$
%\end{theorem}

The next main result we prove is for average-case circular strings: we show that a random circular string can be recovered using a polynomial number of traces. Formally, we prove the following theorem in Section \ref{AverageCase}.

\begin{theorem}\label{AverageCaseThm}
Let $x$ be an unknown but randomly chosen circular string of length $n$ and let $0 < q < 1$ be the deletion probability of each element. Then, there exists a constant $C_q$ depending only on $q$ such that we can determine $x$ with failure probability at most $n^{-10}$ using $O(n^{C_q})$ traces.
\end{theorem}

The main lemma we need to prove Theorem \ref{AverageCaseThm} is actually a result that is true for worst-case strings. Specifically, we show how to recover the multiset of all consecutive substrings of length $O(\log n)$ using a polynomial number of traces. While this does not guarantee that we can recover an arbitrary circular string, it does allow us to recover what we will call \emph{regular strings}, which we show comprise the majority of circular strings. The following lemma may be of independent interest for studying worst-case strings as well, as it allows one to gain information about all ``consecutive chunks'' of the unknown string using only a polynomial number of queries.

\begin{lemma} \label{Contiguous}
    Let $x = x_1 \cdots x_n$ be an arbitrary circular string of length $n$ and let $0 < q < 1$ be the deletion probability of each element. Then, for $k = 100 \log n,$ we can recover the multiset of all substrings $\{x_i x_{i+1} \cdots x_{i+k-1}\}_{i = 1}^{n}$, where indices are modulo $n$, using $O(n^{C_q})$ traces with failure probability $n^{-10}$, where $C_q$ is a constant that only depends on $q$.
\end{lemma}

The best known upper bound for average-case linear trace reconstruction is $\exp\left(O((\log n)^{1/3})\right)$ \cite{HoldenPP18}. Unfortunately, we were not able to adapt their argument to circular strings. One major reason why we are unable to do so is that in the argument of \cite{HoldenPP18} (as well as \cite{PeresZ17}, which provides an $\exp\left(O((\log n)^{1/2})\right)$ sample algorithm), the authors recover the $(k+1)^{\text{st}}$ bit of the string assuming the first $k$ bits are known using a small number of traces, and by reusing traces, they inductively recover the full string. However, since we are dealing with circular strings, even recovering the ``first'' bit does not make much sense. However, we note that even a polynomial-sample algorithm is quite nontrivial. In the linear case, a polynomial-sample algorithm for average-case strings was first proven by \cite{HolensteinMPW08}, and their algorithm only worked as long as the deletion probability $q$ was at most some small constant, which when optimized is only about $0.07$ \cite{PeresZ17}.

\medskip

Our final main result regards lower bounds for worst-case strings. For linear worst-case strings, the best known lower bound for trace reconstruction is $\tilde{\Omega}(n^{3/2})$ \cite{Chase19}. For circular trace reconstruction, we show an improved lower bound of $\tilde{\Omega}(n^3)$. Moreover, the proof of our lower bound is actually much simpler and cleaner than those of the known lower bounds for standard trace reconstruction \cite{Chase19, HoldenL18}. Specifically, we prove the following theorem, done in Section \ref{LowerBound}:

\begin{theorem} \label{LowerThm}
    Let $x$ be the string $10^{n}10^{n+1}10^{n+k} = 1 \underbrace{0 \dots 0}_{n \text{ times}} 1 \underbrace{0 \dots 0}_{n+1 \text{ times}} 1 \underbrace{0 \dots 0}_{n+k \text{ times}},$ where $n \ge 1$ and $2 \le k \le 4$. Likewise, let $y$ be the string $y = 10^{n}10^{n+k}10^{n+1}.$ Then, the strings $x, y$ are not equivalent up to cyclic rotations, but for any constant deletion probability $q$, one requires $\Omega(n^3/\log^3n)$ random traces to distinguish between the original string being $x$ or $y$. Thus, for all integers $n$, worst-case circular trace reconstruction requires at least $\tilde{\Omega}(n^3)$ random traces.
\end{theorem}

\subsubsection{Concurrent Work}

We note that a very similar statement to Lemma \ref{Contiguous}, but for linear strings, was proven in independent concurrent work by Chen et. al. \cite[Theorem 2]{ChenDLSS20}, which provides a polynomial-sample algorithm for a ``smoothed'' variant of worst-case linear trace reconstruction. Many ideas in our proof of Lemma \ref{Contiguous} and their proof appear to overlap, though our proof is substantially shorter. We discuss the relation between our work and \cite{ChenDLSS20} further at the end of Section \ref{AverageCase}.
%While in this paper we focus on sample complexity, we note that there is also a polynomial-\emph{time} algorithm that can recover the multiset of strings in Lemma \ref{Contiguous}, which also implies a polynomial-time algorithm for average-case circular trace reconstruction. However, we defer the runtime bound to Appendix \ref{Runtime} and only give a brief outline, since our ideas for a fast runtime are highly based on those of \cite[Lemma 5.8]{Narayanan20} (and also appear to be nearly identical to those of Chen et. al. \cite{ChenDLSS20}).

\subsection{Motivation and Relation to Other Work}

From a theoretical perspective, circular trace reconstruction can bring many novel insights to the theory of reconstruction algorithms, some of which may be useful even in the standard trace reconstruction problem. For instance, the proof of Theorem \ref{WorstCaseThm} combines analytic, statistical, and combinatorial approaches as in previous trace reconstruction papers, but now also uses ideas from number theory and results about cyclotomic integers. To the best of our knowledge, this paper is the first paper on trace reconstruction to utilize number theoretic ideas, though there is work on other problems about cyclic strings that uses ideas from number theory. Also, Lemma \ref{Contiguous} shows a way to recover all contiguous sequences in the original string of length $O(\log n)$ for arbitrary circular strings, which is a new result even in the linear case (concurrent with \cite{ChenDLSS20}) and has applications to problems in linear trace reconstruction as well (as done in \cite{ChenDLSS20}).

From an applications perspective, trace reconstruction is closely related to the multiple sequence alignment problem in computational biology. In the multiple sequence alignment problem, one is given DNA sequences from several related organisms, and the goal is to align the sequences to determine what mutations each descendant underwent from their common ancestor: the trace reconstruction problem is analogous to actually recovering the common ancestor. One can learn more about trace reconstruction's relation to the multiple sequence alignment problem (as well as to various other problems in computational biology) via the recent survey \cite{BioSurvey}.

The multiple sequence alignment problem is also a key motivation for studying circular trace reconstruction. Many important types of DNA, such as mitochondrial DNA in humans and other eukaryotes, chloroplast DNA, bacterial DNA, and DNA in plasmids, are predominantly circular (see, e.g., \cite[pp. 313, 397, 516-517]{Campbell}, or \cite{Wikipedia}). Therefore, understanding circular trace reconstruction could prove useful in reconstructing ancestral sequences for mitochondrial or bacterial DNA. Another problem in computational biology that trace reconstruction may be applicable to is the DNA Data Storage problem, where data is stored in DNA and can be recovered through sequencing, though the stored DNA may mutate over time \cite{ChurchGK12, OrganickAC18}. Recently, long-term DNA data storage in plasmids has been successfully researched \cite{NguyenPP18}, which further motivates the study of circular trace reconstruction.

\medskip

Besides the linear trace reconstruction problem, circular trace reconstruction is also closely related to the problem of population recovery from the deletion channel \cite{BanCFSS19, BanCSS19, Narayanan20}, where the goal is to recover an unknown mixture of $\ell$ strings from random traces. Indeed, receiving traces from a circular string is equivalent to receiving traces from a uniform mixture of a linear string along with all of its cyclic shifts, so circular trace reconstruction can be thought of as an instance of population recovery from the deletion channel with mixture size $\ell = n$.

Unfortunately, the best known algorithm for population recovery over worst-case strings requires $\exp\left(\tilde{O}(n^{1/3}) \cdot \ell^2\right)$ traces \cite{Narayanan20}, which is not useful if $\ell = n$. However, to prove our worst-case upper bound, we will use ideas based on \cite{DeOS17, NazarovP17, Narayanan20} to estimate certain polynomials that depend on the unknown circular string $x$. For the average case problem, i.e., if given a mixture over $\ell$ random strings, population recovery can be done with $\poly\left(\ell, \exp\left((\log n)^{1/3})\right)\right)$ random traces. While this seemingly implies a $\poly(n)$-sample algorithm for average-case circular trace reconstruction, the $n$ cyclic shifts of the circular string are quite similar to each other and thus do not behave like a collection of $n$ independent random strings. Indeed, our techniques for average-case circular trace reconstruction are very different from those developed in \cite{BanCSS19}.

While circular strings have not been studied before in the context of trace reconstruction, people have studied circular strings and cyclic shifts in the context of edit distance \cite{Maes90, AndoniGMP13}, multi-reference alignment \cite{BandeiraCSZ14, BandeiraWR19, PerryWBRS19}, and other pattern matching problems \cite{CKPRRSWZ21}. We note that \cite{AndoniGMP13} also applies results from number theory and about cyclotomic polynomials, though their techniques overall are not very similar to ours.

\subsection{Proof Overview}

In this subsection, we highlight some of the ideas used in Theorems \ref{WorstCaseThm}, \ref{AverageCaseThm}, and \ref{LowerThm}.

The proof of Theorem \ref{WorstCaseThm} is partially based on ideas from \cite{DeOS17, NazarovP17, Narayanan20}. The authors of \cite{DeOS17, NazarovP17} consider two strings $x, y \in \{0, 1\}^n$ and show how to distinguish between random traces of $x$ and random traces of $y$. To do so, they construct an unbiased estimator for the polynomial $P(z; x) := \sum_{i = 1}^{n} x_i z^i$ (or $P(z; y) = \sum_{i = 1}^{n} y_i z^i$) solely based on the random trace of either $x$ or $y$, for some $z \in \BC$. By showing that the unbiased estimator is never ``too'' large and that $P(z; x)$ and $P(z; y)$ differ enough for an appropriate choice of $z$, they can estimate this quantity using random traces to distinguish between $x$ and $y$. In our case, applying the same estimator will give us an unbiased estimator for $P'(z; x) := \BE_i[P(z; x^{(i)})],$ where $x^{(i)}$ is the $i$th cyclic shift of $x$. Unfortunately, it turns out that $P'(z; x) = P'(z; y)$ as polynomials in $z$ as long as $x, y$ have the same number of $1$'s, even if $x$ and $y$ are vastly different as circular strings. Our goal will then be to establish some other polynomial $Q(z; x)$ such that we can construct a good unbiased estimator, but at the same time $Q'(z; x) := \BE_i[Q(z; x^{(i)})]$ and $Q'(z; y) := \BE_i[Q(z; y^{(i)})]$ are distinct polynomials for any distinct cyclic strings $x, y$. We show that the polynomial $Q(z; x) := z^{kn} P(z; x)^k P(z^{-k}; x)$ will do the job, for some some small integer $k$. We provide a (significantly more complicated) unbiased estimator of $Q(z; x)$ using a random trace: the construction is similar to that of \cite{Narayanan20}, which shows how to estimate $P(z; x)^k$ for some integer $k$. To show that $Q(z; x) \neq Q(z; y)$ as polynomials, we first show that $Q(z; x)$ has the special property that if $z$ is a cyclotomic $n$th root of unity, this polynomial is invariant under cyclic shifts of $x$! Thus, it just suffices to show that if $x, y \in \{0, 1\}^n$ are not cyclic shifts of each other, there is some $n$th root of unity $z = e^{2 \pi i r/n}$ for some $0 \le r \le n-1$ such that $P(z; x)^k P(z^{-k}; x) \neq P(z; y)^k P(z^{-k}; y).$ This will require significant number theoretic computation, and will be true as long as $n$ is a prime or a product of two primes.

The bulk of the proof of Theorem \ref{AverageCaseThm} will be proving Lemma \ref{Contiguous}, which reconstructs all consecutive substrings of length $100 \log n$ in the unknown circular string $x$. For a random string $x$, these substrings will all be sufficiently different, so once we know the substrings, we can reconstruct the full string because there will only be one way to ``glue'' together the substrings. Therefore, we focus on explaining the ideas for Lemma \ref{Contiguous}. Our goal will be to determine how many times a string $s$ appears consecutively in $x$ for each string $s$ of length $100 \log n$. For an unknown string $x$ and $i$ between $0$ and $n- 100 \log n,$ we let $c_i$ be the number of times $s$ appears (possibly non-contiguously) in some contiguous block of length $i + 100 \log n$ in $x$. Then, a basic enumerative argument shows that for a random (cyclically shifted) trace $\tilde{x} = \tilde{x}_1 \tilde{x}_2 \cdots \tilde{x}_m,$ the probability that $\tilde{x}_1 \cdots \tilde{x}_{100 \log n} = s$ can be written as $\sum_{i \ge 0} c_i (1-q)^{100 \log n} q^i,$ and we wish to recover $c_0$. The $(1-q)^{100 \log n}$ term is a constant that equals $1/\poly(n),$ so it is easy to recover an approximation to $\sum_{i \ge 0} c_i q^i$. We truncate this polynomial at an appropriate degree (approximately $C \log n$ for some large $C$) and show that the truncated polynomial $\sum_{i = 0}^{C \log n} c_i x^i$ is very close to the original polynomial, but differs from $\sum_{i = 0}^{C \log n} c_i' x^i$ for some $x \in [q, (1-q)/2]$ by a significant amount, if $c_0' \neq c_0,$ using ideas based on \cite{BorweinEK99}. We can also simulate a trace with deletion probability $x > q$ by taking a ``trace of the trace.'' This will be sufficient in determining $c_0$, and therefore, the (multi)-set of all consecutive substrings of length $100 \log n$.

The proof of Theorem \ref{LowerThm} proceeds by showing that the laws of the traces of $x=10^n10^{n+1}10^{n+k}$ and $y=10^n10^{n+k}10^{n+1}$ are close to each other in the sense of Hellinger distance and concluding by a lemma in \cite{HoldenL18} that was used in a similar fashion to show a lower bound for linear trace reconstruction. It is first shown that conditioned on a $1$ being deleted, a trace from $x$ is equidistributed as a trace from $y$. Then explicit expressions for the probabilities that the trace is $10^a10^b10^c$ are computed and compared, yielding an upper bound on the Hellinger distance. The difference between the probabilities for $x$ and $y$ is proportional to the product of $(a-b)(b-c)(a-c)$ and a symmetric polynomial in $a,b,c$. Both $x$ and $y$ consist of three $1$'s separated by runs of $0$'s of approximate length $n$, so with high probability we have that $a,b,c$ are approximately $np$, with square root fluctuations. The contribution of the $(a-b)(b-c)(a-c)$ term allows us to recover a $\tilde\Omega(n^3)$ bound.

\subsection{Outline}

In Section \ref{Prelim}, we go over some preliminary definitions and results. In Section \ref{WorstCase}, we prove Theorem \ref{WorstCaseThm}. In Section \ref{AverageCase}, we prove Theorem \ref{AverageCaseThm}. In Section \ref{LowerBound}, we prove Theorem \ref{LowerThm}. Finally, in Section \ref{Conclusion}, we conclude by discussing open problems and avenues for further research. Some proofs, such as the proof of Proposition \ref{CircularHarderThanLinear} and the full proof of Theorem \ref{NT:Main}, are deferred to Appendix \ref{Omitted}.

\section{Preliminaries} \label{Prelim}

%We will need a simple result about complex numbers, as well as a ``Littlewood-type'' result about bounds on polynomials on arcs of the unit circle \cite{BorweinE97}. We note that the latter result requires complex analysis, though we will not have to use any knowledge of complex analysis besides this result as a black box. Finally, we will need a well-known result from the theory of cyclotomic fields.

First, we explain a basic definition we will use involving complex numbers.

\begin{definition}
    For $z \in \BC,$ let $|z|$ be the \emph{magnitude} of $z$, and if $z \neq 0$, let $\arg z$ be the \emph{argument} of $z$, which is the value of $\theta \in (-\pi, \pi]$ such that $\frac{z}{|z|} = e^{i \theta}$.
\end{definition}

Next, we state a Littlewood-type result about bounding polynomials on arcs of the unit circle.

\begin{theorem} \textup{\cite{BorweinE97}} \label{Littlewood}
    Let $f(z) = \sum_{j = 0}^{n} a_j z^j$ be a nonzero polynomial of degree $n$ with complex coefficients. Suppose there is some positive integer $M$ such that $|a_0| \ge 1$ and $|a_j| \le M$ for all $0 \le j \le n$. Then, if $A$ is an arc of the unit circle $\{z \in \BC: |z| = 1\}$ with length $0 < a < 2 \pi$, there exists some absolute constant $c_1 > 0$ such that
\[\sup\limits_{z \in A} |f(z)| \ge \exp\left(\frac{-c_1 (1 + \log M)}{a}\right).\]
\end{theorem}

%As a corollary, we prove a bivariate version of the above Littlewood-type result. This generalizes a result proven in \cite{}.

%\begin{theorem} \label{BivariateLittlewood}
    %Let $f(z_1, z_2) = \sum_{j, k = 0}^{n} a_{j, k} z_1^j z_2^k$ be a nonzero bivariate polynomial of degree $n$ with complex coefficients. Suppose there is some positive real number $M$ such that $|a_{j, k}| \le M$ for all $0 \le j, k \le n$ and $|a_{j, k}| \ge 1$ whenever $a_{j, k} \neq 0$. Let $A$ be a subarc of the unit circle $|z| = 1$ in the complex plane with length $0 < a < 2 \pi$. Then, there exists some absolute constant $c_2 > 0$ such that
%\[\sup\limits_{z_1, z_2 \in A} |f(z_1, z_2)| \ge \exp\left(\frac{-c_2 (1 + \log (n \cdot M))}{a^2}\right).\]
%\end{theorem}
%
%We defer the proof of the above theorem to Appendix \ref{Omitted}.

Next, we state two well known results about roots of unity in cyclotomic fields.

\begin{lemma} \label{GaloisConjugate} \cite{Marcus}
    Let $\omega = e^{2 \pi i/n}$. Then, the set of $\{\omega^k\}$ for $k \in \BZ, \gcd(k, n) = 1$ are all \emph{Galois conjugates}. This means that if $P(x)$ is an integer polynomial, then $P(\omega^k) = 0$ if and only if $P(\omega) = 0$ for any $k \in \BZ$ with $\gcd(k, n) = 1.$ Moreover, $P(\omega) = 0$ if and only if $P$ is a multiple of the $n$th \emph{Cyclotomic polynomial}.
\end{lemma}

\begin{lemma} \label{RootsOfUnity} \cite{Marcus}
    Let $\omega = e^{2 \pi i/n}$ be an $n$th root of unity, and let $\BQ[\omega]$ be the $n$th degree cyclotomic field. Then, if $z \in \BQ[\omega]$ is such that $z^r = 1$ for some integer $r \ge 1,$ $z$ must equal $\omega^k$ or $-\omega^k$ for some integer $k$.
\end{lemma}

Finally, we define the Hellinger distance between two probability measures and state a folklore bound on distinguishing between distributions based on samples in terms of the Hellinger distance.

\begin{definition}
    Let $\mu$ and $\nu$ be discrete probability measures over some set $\Omega.$ In other words, for $x \in \Omega$, $\mu(x)$ is the probability of selecting $x $ when drawing from the measure $\mu.$ Then, the Hellinger distance is defined as 
\[d_H(\mu, \nu) = \left(\sum_{x \in \Omega} \left(\sqrt{\mu(x)}-\sqrt{\nu(x)}\right)^2\right)^{1/2}.\]
\end{definition}

The following proposition is quite well-known (see for instance, \cite[Lemma A.5]{HoldenL18}).

\begin{proposition}
    If $\mu, \nu$ are discrete probability measures, then if given i.i.d. samples from either $\mu$ or $\nu$, one must see at least $\Omega(d_H(\mu, \nu)^{-2})$ i.i.d. samples to determine whether the distribution is $\mu$ or $\nu$ with at least $2/3$ success probability.
\end{proposition}

\section{Worst Case: Upper Bound} \label{WorstCase}

%In subsection \ref{Main1Proof}, we prove Theorem \ref{WorstCaseThm}, and in subsection \ref{Main2Proof}, we prove Theorem \ref{Main2}. We note that many of the ideas in our proof for Theorem \ref{WorstCaseThm} will be crucial to proving Theorem \ref{Main2}.

%\subsection{} \label{Main1Proof}

In this section, we prove Theorem \ref{WorstCaseThm}, i.e., we provide an $\exp\left(\tilde{O}(n^{1/3})\right)$-sample algorithm for circular trace reconstruction when the length $n$ is a prime or product of two primes.

For a (linear) string $x \in \{0, 1\}^n$ and $z \in \BC,$ we define $P(z; x) := \sum_{i = 1}^{n} x_i z^i.$ The first lemma we require creates an unbiased estimator for $\prod_{i = 1}^{m} P(z_i; x)$ for some complex numbers $z_1, \dots, z_m,$ using only random traces of $x$. The proof of the following lemma greatly resembles the proof of \cite[Lemma 4.1]{Narayanan20}, so we defer the proof to Appendix \ref{Omitted}.

\begin{lemma} \label{CreatingG}
    Let $x$ be a linear string of length $n$. Fix $q$ as the deletion probability and $p = 1-q$ as the retention probability. Then, for any integer $m \ge 1$ and any $Z = (z_1, \dots, z_m)$ for $z_1, \dots, z_m \in \BC,$ there exists some function $g_m(\tilde{x}, Z)$ such that
\[\BE_{\tilde{x}}[g_m(\tilde{x}, Z)] = \prod_{k = 1}^{m} \left(\sum_{i = 1}^{n} x_i z_k^i\right),\]
    where the expectation is over traces drawn from $x$. Moreover, for any $L \ge 1,$ and for all $\tilde{x} \in \{0, 1\}^n$ and all $Z$ such that $|z_1|, \dots, |z_m| = 1$ and $|\arg z_i| \le \frac{1}{L}$ for all $1 \le i \le m,$ 
\[|g_m(\tilde{x}, Z)| \le (p^{-1} m n)^{O(m)} \cdot e^{O(m^2n/(p^2L^2))}.\]
%    Finally, $g_m(\tilde{x}, Z)$ can be computed in $n^{O(m)}$ time.
\end{lemma}

For $x \in \{0, 1\}^n$ and $z \in \BC,$ let $P(z; x) := \sum_{i = 1}^{n} x_i z^i.$ Our main goal will be to determine the value of $f_t(z; x) := P(z; x)^t \cdot P(z^{-t}; x)$ for some integer $t$, where $z$ is an $n$th root of unity. Importantly, we note that $f_t(z; x)$ is invariant under rotations of $x$, since for $z = e^{2 \pi i k/n},$
\[\sum_{i = 1}^{n} x_{(i+1) \Mod n} z^i = \sum x_i z^{i-1} = P(z; x) \cdot z^{-1}\]
whereas
\[\sum_{i = 1}^{n} x_{(i+1) \Mod n} z^{-t \cdot i} = \sum x_i z^{-t(i-1)} = P(z^{-t}; x) \cdot z^{t}\]
Therefore, if we define $x^{(j)}$ as the string $x$ rotated by $j$ places (so $x^{(j)}_i = x_{(i+j) \Mod n}$), then $f(z; x) = f(z; x^{(j)})$ for all $z = e^{2 \pi i k/n}$ and $0 \le j \le n-1.$

Now, choose some $z$ with $|z| = 1$ and $|\arg z| \le \frac{1}{L}$. Also, fix some integer $t$, let $m = t+1,$ and let $Z = ( \underbrace{z, \dots, z}_{t \text{ times}}, z^{-t}).$ Then, if $j$ is randomly chosen in $\{0, 1, \dots, n-1\}$ and $\tilde{x}$ is a random trace,
\[\BE_{\tilde{x}}[n z^{tn} \cdot g_m(\tilde{x}, Z)] = (n \cdot z^{tn}) \cdot \left(\frac{1}{n} \cdot \sum_{j = 0}^{n-1} P(z; x^{(j)})^t \cdot P(z^{-t}; x^{(j)})\right) = \sum_{j = 0}^{n-1} z^{tn} \cdot P(z; x^{(j)})^t \cdot P(z^{-t}; x^{(j)}),\]
where $g_m(\tilde{x}, Z)$ is defined in Lemma \ref{CreatingG}.
Note that $\sum_{j = 0}^{n-1} z^{t n} \cdot P(z; x^{(j)})^t \cdot P(z^{-t}; x^{(j)})$ is a polynomial of $z$ of degree at most $3tn$ and all coefficients bounded by $n^{t+1}$. We write this polynomial as $Q_t(z; x).$ Thus, if we define $h_t(\tilde{x}, z) := n z^{tn} g_{m}(\tilde{x}, Z),$ we have that $\BE_{\tilde{x}}[h_t(\tilde{x}, z)] = Q(z; x)$ for $\tilde{x}$ a trace of a randomly shifted $x$, and that $|h_t(\tilde{x}; z)| \le (p^{-1} t n)^{O(t)} \cdot e^{O(t^2n/(p^2L^2))}$ whenever $|z| = 1$ and $|\arg z| \le \frac{1}{L}$ for $L \ge 2$ and $m = t+1$, by Lemma \ref{CreatingG}.

Now, we will state two important results that will lead to the proof of the main result.

\begin{lemma} \label{Arc}
    Let $n \ge 2,$ and suppose that $x, x'$ are strings in $\{0, 1\}^n$ such that $Q_t(z; x) \neq Q_t(z; x')$ as polynomials in $z.$ Then, there is an absolute constant $c_2$ such that for any $L \ge 2,$ there exists $z$ such that $|z| = 1,$ $|\arg z| \le \frac{1}{L},$ and
\[|Q_t(z; x) - Q_t(z; x')| \ge n^{-c_2 t L}.\]
\end{lemma}

\begin{proof}
    Note that $Q_t(z; x) - Q_t(z; x')$ is a nonzero polynomial in $z$ of degree at most $(t+1)n$ and with all coefficients bounded by $2n^{t+1}.$ Therefore, by Theorem \ref{Littlewood},
\[\sup_{|z| = 1, |\arg z| \le 1/L} |Q_t(z; x) - Q_t(z; x')| \ge \exp\left(-\frac{c_1 (1 + \log (2n^{t+1}))}{2/L}\right) \ge \exp\left(-c_2 \cdot L \cdot t \cdot \log n\right) = n^{-c_2 t L},\]
    where we note that the arc $\{z: |z| = 1, |\arg z| \le \frac{1}{L}\}$ has length $\frac{2}{L}.$
\end{proof}

The next important result we need will be Theorem \ref{NT:Main}. We defer the full proof of Theorem \ref{NT:Main} to Subsection \ref{NT:Subsection}, but as the proof of the case where $n$ is prime is simpler, we prove this special case here. Using this, we can get an $\exp\left(\tilde{O}(n^{1/3})\right)$ sample upper bound at least for $n$ prime. As a note, we will define $\omega = e^{2 \pi i/n}$ from now on, where $n$ will be clear from context.

\begin{proposition} \label{NT:Prime}
    Suppose that $n = p$ is prime, and $a_0, \dots, a_{p-1}, b_0, \dots, b_{p-1} \in \{0, 1\}$ are such that for all $0 \le k < p,$ there is some integer $c_k$ such that $\sum_{i = 0}^{p-1} a_i = \omega^{c_k} \cdot \sum_{i = 0}^{p-1} b_i.$ Then, the sequences $\{a_1, \dots, a_p\}$ and $\{b_1, \dots, b_p\}$ are equivalent up to a cyclic permutation.
\end{proposition}

\begin{proof}
    First, $\sum_{i = 0}^{p-1} a_i = \omega^{c_0} \cdot \sum_{i = 0}^{p} b_i$. Since $\sum_{i = 0}^{p-1} a_i$ and $\sum_{i = 0}^{p-1} b_i$ are both nonnegative real numbers, and since $\omega^{c_0}$ is a root of unity, we must have that $\sum_{i = 0}^{p-1} a_i = \sum_{i = 0}^{p-1} b_i$. 
    %In the case $p = 2,$ this alone proves the proposition, so we now assume $p$ is odd.
    
    Next, we have that $\sum_{i = 0}^{p-1} a_i \omega^i = \omega^{c_1} \cdot \sum_{i = 0}^{p-1} b_i \omega^i$. Letting $b'_i = b_{(i-c_1) \Mod p}$, we have that $b'$ is a cyclic shift of $b$, and $\sum_{i = 0}^{p-1} a_i = \sum_{i = 0}^{p-1} b'_i$ and $\sum_{i = 0}^{p-1} a_i \omega^i = \sum_{i = 0}^{p-1} b'_i \omega^i.$ Letting $Q(x) = \sum_{i = 0}^{p-1} (a_i-b'_i) x^i,$ we have that $\omega$ and $1$ are both roots of $Q(x)$. Since $Q(x)$ is an integer-valued polynomial, this implies that all Galois conjugates of $\omega$ are roots, so $1, \omega, \omega^2, \dots, \omega^{p-1}$ are roots of $Q(x).$ Thus, $x^p-1$ divides $Q(x).$ But since $Q(x)$ has degree at most $p-1,$ $Q(x)$ must equal $0$, so $a_i = b'_i$ for all $i$. Since the sequence $b'$ is just a shift of $b$, we are done.
\end{proof}

By using Theorem \ref{NT:Main} (or Proposition \ref{NT:Prime} in the case of $n$ prime), we obtain the following number theoretic result.

\begin{lemma} \label{PolysAreDifferent}
    Let $n$ be a prime or a product of two primes, and let $a = a_1 a_2 \cdots a_n$ and $b = b_1 b_2 \cdots b_n$ be distinct $n$-bit strings (even up to cyclic shift). Then, for some $0 \le \ell \le n-1$ with $z = \omega^{\ell}$, and for some $2 \le t \le 5$, we have that $P(z; a)^t P(z^{-t}; a) \neq P(z; b)^t P(z^{-t}; b).$
\end{lemma}

\begin{proof}
    First choose $k$ such that $\sum_{i = 1}^{n} a_i \omega^{i \cdot k} \neq \omega^{c_k} \cdot \sum_{i = 1}^{n} b_i \omega^{i \cdot k}$ for all integers $c_k,$ which exists by Theorem \ref{NT:Main}. If $k = 0,$ then $P(\omega^k; a) = P(1; a)$ and $P(\omega^k; b) = P(1; b)$ are distinct nonnegative integers, so we trivially have $P(1; a)^t P(1; a) \neq P(1; b)^t P(1; b).$ Otherwise, let $t$ be the smallest prime that doesn't divide $\frac{n}{\gcd(n, k)}$ (so $t \le 5$ as $n$ has at most $2$ prime factors). If $\sum_{i = 1}^{n} a_i \omega^{i \cdot k} = 0,$ then $\sum_{i = 1}^{n} b_i \omega^{i \cdot k} \neq 0.$ Now, since $\omega^{-tk}$ is a Galois conjugate of $\omega^k$ (since $t \nmid n$), we also have that $\sum_{i = 1}^{n} b_i \omega^{-ti \cdot k} \neq 0.$ This means that $P(\omega^k; a) = 0$ so $P(\omega^k; a)^t P((\omega^k)^{-t}; a) = 0,$ but $P(\omega^k; b)^t P((\omega^k)^{-t}; b) \neq 0$. Likewise, if $\sum_{i = 1}^{n} b_i \omega^{i \cdot k} = 0,$ we'll have $P(\omega^k; a)^t P((\omega^k)^{-t}; a) \neq 0,$ but $P(\omega^k; b)^t P((\omega^k)^{-t}; b) = 0$. This means the result follows if either $P(\omega^k; a) = 0$ or $P(\omega^k, b) = 0$.
    
    Otherwise, $P(\omega^k; a) = \sum_{i = 1}^{n} a_i \omega^{i \cdot k}$ and $P(\omega^k; b) = \sum_{i = 1}^{n} b_i \omega^{i \cdot k}$ are both nonzero. This means that for all $r \ge 0,$ $P(\omega^{(-t)^r \cdot k}; a)$ and $P(\omega^{(-t)^r \cdot k}; b)$ are both nonzero, since $\omega^{(-t)^r \cdot k}$ and $\omega^k$ are Galois conjugates. This means that if $P(z; a)^t P(z^{-t}; a) = P(t; b)^t P(z^{-t}; b)$ for all $z = \omega^{(-t)^r \cdot k},$ then 
\[\frac{P(\omega^{(-t)^{r+1} \cdot k}; a)}{P(\omega^{(-t)^r \cdot k}; a)^{-t}} = \frac{P(z^{-t}; a)}{P(z; a)^{-t}} = \frac{P(z^{-t}; b)}{P(z; b)^{-t}} = \frac{P(\omega^{(-t)^{r+1} \cdot k}; b)}{P(\omega^{(-t)^r \cdot k}; b)^{-t}}\]
    for all $r \ge 0,$ so we inductively have that
\[\frac{P(\omega^{(-t)^r \cdot k}; a)}{P(\omega^k; a)^{(-t)^r}} = \frac{P(\omega^{(-t)^r \cdot k}; b)}{P(\omega^k; b)^{(-t)^r}}.\]
    Now, letting $r = \varphi\left(\frac{n}{\gcd(n, k)}\right),$ we know that $k \cdot (-t)^r \equiv k \Mod n$ by Euler's theorem, which means that $\omega^{(-t)^r \cdot k} = \omega^k.$ Thus,
\[P(\omega^k; a)^{1 - (-t)^r} = P(\omega^k; b)^{1 - (-t)^r}.\]
    Since $k \neq 0,$ we have that $\frac{n}{\gcd(n, k)} > 1$ so $r \ge 1.$ Thus, since $t \ge 2,$ $1-(-t)^r \neq 0.$ Now, since $P(\omega^k; a), P(\omega^k; b)$ are nonzero, we have that $\frac{P(\omega^k; a)}{P(\omega^k; b)}$ is a $|1-(-t)^r|^{\text{th}}$ root of unity. Also, $P(\omega^k; a), P(\omega^k; b) \in \BQ[\omega],$ which means $\frac{P(\omega^k; a)}{P(\omega^k; b)} \in \BQ[\omega]$. However, all roots of unity in $\BQ[\omega]$ are of the form $\pm \omega^i$ for some $i$, and since $(-t)^{r}-1$ is odd if $n$ is odd (since $t = 2$), we must have that $\frac{P(\omega^k; a)}{P(\omega^k; b)} = \omega^{c_k}$ for some integer $c_k.$ This is a contradiction, so we must have that $P(z; a)^t P(z^{-t}; a) \neq P(z; b)^t P(z^{-t}; b),$ for some $z = \omega^{(-t)^r \cdot k}$, $r \ge 0$.
\end{proof}

Finally, we are ready to prove Theorem \ref{WorstCaseThm}.

\begin{proof}[Proof of Theorem \ref{WorstCaseThm}]
    Suppose that we are trying to distinguish between the original circular string being $a = a_1 a_2 \cdots a_n$ or $b = b_1 b_2 \cdots b_n,$ where $a, b$ are distinct, even up to cyclic shifts.

    We choose $\ell, t$ based on Lemma \ref{PolysAreDifferent}, so that $P(\omega^\ell; a)^t P((\omega^\ell)^{-t}; a) \neq P(\omega^\ell; b)^t P((\omega^\ell)^{-t}; b)$. As we have already noted, if $z$ is an $n^\text{th}$ root of unity, then $P(z; a)^t P(z^{-t}; a)$ is invariant under rotation of $a$, and $P(z; b)^t P(z^{-t}; b)$ is invariant under rotation of $b$. By our definition of $Q_t(z; x),$ we have that $Q_t(\omega^\ell; a) \neq Q_t(\omega^\ell; b),$ so $Q_t(z; a) \neq Q_t(z; b)$ as polynomials in $z$. Therefore, by Lemma \ref{Arc}, there is some $z$ such that $|z| = 1, |\arg z| \le \frac{1}{L}$, and
\[|Q_t(z; a) - Q_t(z; b)| \ge n^{-c_2 t L} \ge n^{-5 c_2 L}.\]
    So, for $L = \left\lceil n^{1/3} (\log n)^{-1/3} p^{-2/3} \right\rceil,$ there exists $z$ with $|z| = 1$ and $|\arg z| \le \frac{1}{L}$ and some $2 \le t \le 5$ such that
\[|Q_t(z; a) - Q_t(z; b)| \ge n^{-5 c_2 L} \ge \exp\left(-c_3 \cdot n^{1/3} (\log n)^{2/3} p^{-2/3}\right),\]
    but
\[|h_t(\tilde{x}, z)| \le (p^{-1} n)^{O(1)} \cdot \exp\left(O\left(\frac{n}{p^2 L^2}\right)\right) \le \exp\left(c_4 \cdot n^{1/3} (\log n)^{2/3} p^{-2/3}\right)\]
    for any trace $\tilde{x}$ of either $a$ or $b$. Sample $R = \exp\left(O\left(n^{1/3} (\log n)^{2/3} p^{-2/3}\right)\right)$ traces $\tilde{x}^{(1)}, \dots, \tilde{x}^{(R)}$ and choose $z$ and $t$ based on Lemma \ref{PolysAreDifferent}. Then, if we define $h_t(z)$ to be the average of $h_t(\tilde{x}^{(i)}, z)$ over $i$ from $1$ to $R$, the Chernoff bound tells us that with failure probability at most $10^{-n}$, $|h_t(z) - Q_t(z; a)| \le \frac{1}{3} \cdot \exp\left(c_4 \cdot n^{1/3} (\log n)^{2/3} p^{-2/3}\right)$ if the original string were $a,$ and $|h(z) - Q_t(z; b)| \le \frac{1}{3} \cdot \exp\left(c_4 \cdot n^{1/3} (\log n)^{2/3} p^{-2/3}\right)$ if the original string were $b.$ Thus, by returning $a$ if $h(z)$ is closer to $Q_t(z; a)$ and returning $b$ otherwise, we can distinguish between the original string being $a$ or $b$ using $\exp\left(O\left(n^{1/3} (\log n)^{2/3} p^{-2/3}\right)\right)$ traces, with $1-10^{-n}$ failure probability.
    
    Thus, to reconstruct the original string $x$, we simply run the distinguishing algorithm for all pairs $a, b \in \{0, 1\}^n$ such that $a \neq b,$ using the same $R$ traces $\tilde{x}^1, \dots, \tilde{x}^R.$ With probability at least $1 - (4/10)^n \ge 1 - 2^{-n},$ the true string $x$ will be the only string such that the distinguishing algorithm will successfully choose $x$ over all other strings. Thus, for $n$ a prime or a product of two primes, the circular trace reconstruction problem can be solved using $\exp\left(O\left(n^{1/3} (\log n)^{2/3} p^{-2/3}\right)\right)$ traces.
\end{proof}

\section{Average Case: Upper Bound} \label{AverageCase}
We now consider the situation in which the unknown circular string $x$ is random. For the sake of simplicity, we will assume that we may sample $x$ by sampling a uniform random linear binary string of length $n$ and applying a uniform random cyclic shift. Note that the resulting distribution on $x$ is not uniform over all possible circular strings, as strings with nontrivial cyclic symmetries are more likely to appear. However, such strings form a negligible fraction of all strings, and our arguments can easily be modified to handle the situation in which the distribution is uniform over all possible circular strings. We use the randomness to rule out certain problematic strings with high probability, and this can be done for uniform random circular strings as well as other distributions, for example if in the sampling procedure each bit is independently biased towards $0$ or $1$.

\begin{theorem}\label{AverageCasePoly}
Let $x$ be a random (in the sense described above) unknown circular string of length $n$ and let $q$ be the deletion probability of each element. Then there exists a constant $C_q$ depending only on $q$ such that we can determine $x$ with failure probability at most $n^{-10}$ using $O(n^{C_q})$ traces.
\end{theorem}

In what follows, we will let $x=x_1\cdots x_n$ and take indices of bits in $x$ modulo $n$. Let $k=100\log n$. We first note that with high probability, all of the consecutive substrings of $x$ of length $k$ and $k-1$ are pairwise distinct. We will refer to such strings $x$ as \emph{regular} strings. Indeed, the probability that $x_i\cdots x_{i+k-1}=x_j\cdots x_{j+k-1}$ for $i\ne j$ is $2^{-k}$ (where indices are taken modulo $n$), and union bounding over all $i,j$ as well as both $k$ and $k-1$ gives a failure probability of at most $O(n^2 2^{-k}) \ll n^{-10}$. 

%Thus, we may henceforth assume that $x$ is regular. Note that under this assumption, 
If we assume that $x$ is regular,
the length $k$ consecutive substrings of $x$ uniquely determine $x$. Indeed, given $x_i\cdots x_{i+k-1}$, we can uniquely determine $x_{i+k}$ as there is a unique length $k$ consecutive substring of $x$ that begins with $x_{i+1}\cdots x_{i+k-1}$. Iteratively applying this allows us to recover the entire string $x$.
Thus, to prove Theorem \ref{AverageCasePoly}, it suffices to prove Lemma \ref{Contiguous}, i.e., to determine how many times each length $k$ substring appears consecutively in any string $x$ using $O(n^{C_q})$ traces, which will allow us to recover $x$ if $x$ is regular.

\begin{proof}[Proof of Lemma \ref{Contiguous}]

%To prove Theorem \ref{AverageCasePoly}, 
For a circular string $x$, let $S_x=\{x_ix_{i+1}\cdots x_{i+k-1}\}_{i=1}^n$ be the \emph{$k$-deck} of $x$, which is the multiset of contiguous substrings of $x$ of length $k$. Let $S$ and $T$ be the $k$-decks of some circular strings of length $n$ such that $S\ne T$. Suppose that given $O(n^{C_q})$ traces of a circular string, we are able to distinguish between whether its $k$-deck is $S$ or $T$ correctly with failure probability $10^{-n}$. Then, by union bounding over all possible pairs of distinct $k$-decks (note that there are at most $2^n$ different $k$-decks), we can with high probability correctly determine the $k$-deck of the string, showing the lemma.

The key property we will use is that given two distinct $k$-decks $S$ and $T$ that come from circular strings $x$ and $y$, there exists some string $s$ of length $k$ such that the number of consecutive occurrences of $s$ in $x$ and in $y$ are different. We will show the existence of $C_q$ so that for any string $s$ of length $k$ satisfying this property, we can distinguish between $x$ and $y$ correctly using $O(n^{C_q})$ samples with failure probability $10^{-n}$, from which the result follows.

Let $\alpha$ denote a sufficiently large positive integer only depending on $q$ that we will determine later.
%Note that since $x$ and $y$ are different regular strings, there exists a string $s$ of length $k$ that appears consecutively in $x$ but not $y$ by the previous discussion. 
For $0\le i\le n-k$, let $c_i$ denote the number of (not necessarily consecutive) occurrences of $s$ in $x$ contained in a consecutive substring of $x$ of length at most $i+k$. Similarly, let $d_i$ denote the number of (not necessarily consecutive) occurrences of $s$ in $y$ contained in a consecutive substring of $y$ of length at most $i+k$. By assumption, we have that $c_0 \neq d_0$. By casework on the last bit of the occurrence of $s$, we have that $c_i,d_i\le n\binom{i+k}k$. Let $P(t)=\sum_{i=0}^{\alpha k}c_it^i$ and $Q(t)=\sum_{i=0}^{\alpha k}d_it^i$. Moreover, the following is true:

\begin{lemma}\label{ProbStartWa}
The probability that a trace of $x$ starts with $s$ (where a random bit in the string is chosen as the beginning before bits are deleted) is $\frac1n(1-q)^kP(q)+O(q^{\alpha k}(\alpha+1)^ke^k)$. Similarly, the probability that a trace of $y$ starts with $s$ is $\frac1n(1-q)^kQ(q)+O(q^{\alpha k}(\alpha+1)^ke^k)$.
\end{lemma}

\begin{proof}
To compute the probability that a trace of $x$ starts with $s$, we do casework on how many bits are deleted before the last bit in the occurrence of $s$. If $i$ bits are deleted, then note that there are $c_i$ ways for it to be done by definition. Each such way has a probability of $\frac1n(1-q)^kq^i$ to occur. Indeed, for each way there is a $\frac1n$ probability that the correct starting bit is chosen, and the probability that only the bits corresponding to the specific instance of $s$ are kept is $(1-q)^kq^i$. It follows that the probability is exactly $\frac1n(1-q)^k\sum_{i=0}^{n-k}c_iq^i$.

It remains to show that $\frac1n(1-q)^k\sum_{\alpha k+1}^{n-k}c_iq^i=O(q^{\alpha k}(\alpha+1)^ke^k)$. As mentioned before, we have that $c_i\le n\binom{i+k}k$. Thus, this sum is at most $\sum_{i>\alpha k}\binom{i+k}kq^i\le\binom{\alpha k+k}kq^{\alpha k}\sum_{i\ge0}\left(\frac{q(\alpha+1)}\alpha\right)^i$. Indeed, the ratio of consecutive terms in the sequence $\binom{i+k}kq^i$ is equal to $q\frac{i+k}i\le\frac{q(\alpha+1)}\alpha$. For a sufficiently large choice of $\alpha$, $\frac{q(\alpha+1)}\alpha<1$, so $\sum_{i>\alpha k}\binom{i+k}kq^i=O(\binom{\alpha k+k}kq^{\alpha k})=O(q^{\alpha k}(\alpha+1)^ke^k)$ by Stirling's approximation.

The argument for $y$ is analogous.
\end{proof}

Lemma \ref{ProbStartWa} allows us to estimate $P(q)$ and $Q(q)$ up to an $O(n(1-q)^{-k}q^{\alpha k}(\alpha+1)^ke^k)$ error by looking at how often traces of $x$ or $y$ begin with $s$, and then dividing by $\frac{1}{n}(1-q)^k$. So long as $P(q)$ and $Q(q)$ are sufficiently far apart, a Chernoff bound allows us to determine with high probability if the traces came from $x$ or $y$. However, it may be the case that $P(q)$ and $Q(q)$ are quite close. To remedy this, we observe that it is possible to \emph{simulate higher deletion probabilities $q'>q$}. Indeed, this can be achieved by deleting each bit in traces received independently with probability $\frac{q'-q}{1-q}$. Thus, it suffices to find $q'\in[q,r]$ with $P(q')$ and $Q(q')$ far apart for some $q<r<1$. 
%As long as it exists, we can find it by evaluating $P-Q$ for integral multiples of some small positive real number in $[q,r]$, which does not require any traces. 
The existence of such a $q'$ is proven by the following Littlewood-type result of Borwein, Erdélyi, and Kós.

\begin{theorem}[\cite{BorweinEK99}, Theorem 5.1]\label{BEKLittlewood}
There exist absolute constants $c_1>0$ and $c_2>0$ such that if $f$ is a polynomial with coefficients in $[-1,1]$ and $a\in(0,1]$, then \[|f(0)|^{c_1/a}\le\exp\left(\frac{c_2}a\right)\sup_{z\in[1-a,1]}|f(z)|.\]
\end{theorem}

Let $r=\frac{q+1}2$. We first apply Theorem \ref{BEKLittlewood} to $\binom{\alpha k+k}k^{-1}(P(rx)-Q(rx))$ and $a=1-q/r$. Here, we are using the fact that the coefficients of $P$ and $Q$ are bounded in magnitude by $\binom{\alpha k+k}k$ by previous observations, and that $|P(0)-Q(0)| \ge 1$. Theorem \ref{BEKLittlewood} tells us that 
\begin{align*}
    \binom{\alpha k+k}k^{-c_1/a} &\le\exp\left(\frac{c_2}a\right)\binom{\alpha k+k}k^{-1}\sup_{z\in[1-a,1]}|P(rz)-Q(rz)| \\
    &=\exp\left(\frac{c_2}a\right)\binom{\alpha k+k}k^{-1}\sup_{q'\in[q,r]}|P(q')-Q(q')|,
\end{align*} or \[\sup_{q'\in[q,r]}|P(q')-Q(q')|\ge c_3\binom{\alpha k+k}k^{-c_4}\] for some constants $c_3$ and $c_4$ that only depend on $q$.

In particular, this is much larger than $10^kn(1-r)^{-k}r^{\alpha k}(\alpha+1)^ke^k$ for sufficiently large values of $\alpha$ ($\alpha$ may depend on $q$). Indeed, after taking $k$th roots and using Stirling's approximation this reduces to showing that $(e(\alpha+1))^{-c_5}>10 n^{1/k} (1-r)^{-1} r^\alpha(\alpha+1)e$ for sufficiently large $\alpha$ where $c_5$ is some constant that only depends on $q$, which is clear (since $0 < r < 1$ is fixed and $n^{1/k} < 2$). Thus, for any $q' \in [q, r],$ the error term $\frac1n(1-q')^k\sum_{\alpha k+1}^{n-k}c_i(q')^i=O((q')^{\alpha k}(\alpha+1)^ke^k)$ is at most $10^{-k}$ times $\frac{1}{n} (1-q')^{k} \cdot \sup_{q' \in [q, r]} |P(q')-Q(q')|.$

Hence, for some $q'\in[q,r]$, the probability that a trace begins with $s$ under bit deletion with probability $q'$ differs between $x$ and $y$ by $\Omega(10^kn(1-r)^{-k}r^{\alpha k}(\alpha+1)^ke^k)=\Omega(n^{-c_6})$ for some constant $c_6$ that only depends on $q$. By a standard Chernoff bound, for some constant $C_q$ only depending on $q$, we can distinguish between $x$ and $y$ using $O(n^{C_q})$ traces with failure probability at most $\exp(-\Omega(n))$, so Lemma \ref{Contiguous} follows.
\end{proof}

By the previous discussions, Theorem \ref{AverageCaseThm} is also proven.

\medskip

As mentioned before, Chen et. al. independently proved the analogue of Lemma \ref{Contiguous} for linear strings in \cite[Theorem 2]{ChenDLSS20}. The ideas behind their result and ours are similar: both are proved by considering a polynomial encoding the $k$-deck and using complex analysis to bound the corresponding polynomials for different $k$-decks away from each other on $[q,1]$. Here, we directly apply the Littlewood-type result from \cite{BorweinEK99}, while Chen et. al. prove their own result for this bound. This idea of reconstructing the $k$-deck of the unknown string may be useful in other variants of trace reconstruction, though the sample complexity it achieves is only polynomial and in order to apply it, we must be able to reconstruct the string from its $k$-deck. For worst case trace reconstruction, we must be able to reconstruct strings not uniquely determined by their $k$-decks, and for average case trace reconstruction, a subpolynomial sample complexity has already been achieved for constant deletion probabilities. Thus, these ideas are not directly applicable to worst case and average case trace reconstruction, though they may be helpful for variants in which it suffices to construct $k$-decks and a polynomial sample complexity is not known, such as the ones considered in this paper and in \cite{ChenDLSS20}.

One difference between our result and that of Chen et. al. is that while we only addressed the sample complexity, they also give a polynomial time algorithm for reconstructing the string via a linear program. Their approach can be modified to give a polynomial time algorithm for average case circular trace reconstruction as well. For more details on their algorithm, see \cite[Section 6]{ChenDLSS20}. Moreover, while we do not address the smoothed complexity model as in Chen et. al., our proof of Theorem \ref{AverageCaseThm} easily generalizes to a polynomial sample (or time) algorithm for circular trace reconstruction in the smoothed complexity model. This is because one can show that a circular string drawn from Chen et. al.'s smoothed model is regular with very high probability, in a similar way to how we showed an average string is regular. For more details, see \cite[Section 3]{ChenDLSS20}.

\section{Worst Case: Lower Bound}
\label{LowerBound}
In this section, we prove Theorem \ref{LowerThm} and demonstrate that worst-case circular trace reconstruction requires $\tilde\Omega(n^3)$ traces. We first record the following lemma from \cite{HoldenL18} expressing the number of independent samples required to distinguish between two probability measures $\mu$ and $\nu$ in terms of their Hellinger distance $d_H(\mu,\nu)$, defined to be $\left(\sum_{x\in X}\left(\sqrt{\mu(\{x\})}-\sqrt{\nu(\{x\})}\right)^2\right)^{1/2}$ where the sum is over all events in some discrete sample space $X$. Let $d_{TV}(\mu,\nu)$ denote the total variation distance between $\mu$ and $\nu$ and $\mu^n$ denote the law of $n$ independent samples from $\mu$.

\begin{lemma}[\cite{HoldenL18}, Lemma A.5] \label{HLLemma}
If $\mu$ and $\nu$ are probability measures satisfying $d_H(\mu,\nu)\le1/2$, then 
%for $m\ge1/(4d_H^2(\mu,\nu))$, we have that 
$1-d_{TV}(\mu^m,\nu^m) \ge \varepsilon$ if $m\le\frac{\log(1/\varepsilon)}{9d_H^2(\mu,\nu)}$.
\end{lemma}

Given $m$ traces, we cannot determine if they came from $\mu^m$ or $\nu^m$ with probability higher than $\frac{1}{2} \left(1 + d_{TV}(\mu^m,\nu^m)\right)$. Thus, it requires $\Omega(d_H^{-2}(\mu,\nu))$ samples to distinguish between two probability measures $\mu$ and $\nu$ with probability greater than $\frac{3}{4}$.

\begin{proof}[Proof of Theorem \ref{LowerThm}]
We specialize to the case of distinguishing between $x=10^n10^{n+1}10^{n+k}$ and $y=10^n10^{n+k}10^{n+1}$ from independent traces. Let $\mu$ and $\nu$ respectively denote the laws of traces from $x$ and $y$. We will show that $d_H^2(\mu,\nu)=O((\log n/n)^{3})$, which establishes the result by Lemma \ref{HLLemma}.

First, we note that conditional on the first $1$ in $x$ being deleted, the resulting trace is equidistributed as a trace from $y$ conditioned on the second $1$ being deleted, as in both cases we obtain a trace from the circular string $10^{n+1}10^{2n+k}$. Similar arguments for other cases show that conditioned on any $1$ being deleted, traces from $x$ and $y$ are equal in law. Thus, the resulting string must have three $1$'s to contribute to the Hellinger distance. We will henceforth assume that the resulting trace is of the form $10^a10^b10^c$ for some nonnegative integers $a,b,c$.

We now compute the ratio $\frac{\mu(\{10^a10^b10^c\})}{\nu(\{10^a10^b10^c\})}$ and show that it is typically $1+O((\log n/n)^{3/2})$. We have that \[\frac{\mu(\{10^a10^b10^c\})}{q^{3n+k+1-a-b-c}(1-q)^{a+b+c}}=\binom na\binom{n+1}b\binom{n+k}c+\binom nb\binom{n+1}c\binom{n+k}a+\binom nc\binom{n+1}a\binom{n+k}b,\]
\[\frac{\nu(\{10^a10^b10^c\})}{q^{3n+k+1-a-b-c}(1-q)^{a+b+c}}=\binom na\binom{n+k}b\binom{n+1}c+\binom nb\binom{n+k}c\binom{n+1}a+\binom nc\binom{n+k}a\binom{n+1}b.\]
It follows that \[\frac{\mu(\{10^a10^b10^c\})}{\nu(\{10^a10^b10^c\})}=\frac{\frac1{(n+1-b)(n+1-c)\cdots(n+k-c)}+\frac1{(n+1-c)(n+1-a)\cdots(n+k-a)}+\frac1{(n+1-a)(n+1-b)\cdots(n+k-b)}}{\frac1{(n+1-c)(n+1-b)\cdots(n+k-b)}+\frac1{(n+1-a)(n+1-c)\cdots(n+k-c)}+\frac1{(n+1-b)(n+1-a)\cdots(n+k-a)}}.\]
Multiplying the numerator and denominator by $\prod_{i=1}^k(n+i-a)(n+i-b)(n+i-c)$ results in \[S_1=\prod_{i=1}^k(n+i-a)\prod_{i=2}^k(n+i-b)+\prod_{i=1}^k(n+i-b)\prod_{i=2}^k(n+i-c)+\prod_{i=1}^k(n+i-c)\prod_{i=2}^k(n+i-a)\] and \[S_2=\prod_{i=1}^k(n+i-b)\prod_{i=2}^k(n+i-a)+\prod_{i=1}^k(n+i-c)\prod_{i=2}^k(n+i-b)+\prod_{i=1}^k(n+i-a)\prod_{i=2}^k(n+i-c),\] respectively. We have that $S_1-S_2=(a-b)\prod_{i=2}^k(n+i-a)(n+i-b)+(b-c)\prod_{i=2}^k(n+i-b)(n+i-c)+(c-a)\prod_{i=2}^k(n+i-c)(n+i-a)$. This is an alternating polynomial in $a,b,c$, i.e., applying a permutation $\sigma$ to $a,b,c$ changes the sign of the polynomial by the sign of $\sigma$. Hence, it can be written in the form $(a-b)(b-c)(a-c)P_k(n,a,b,c)$, where $P_k$ is a polynomial in $n,a,b,c$ of degree $2k-4$ since $S_1$ and $S_2$ have degree $2k-1$.

By a standard Chernoff bound, there exists a constant $C$ such that with probability at least $1-n^{-100}$, $a,b,c\in[np-C\sqrt{n\log n},np+C\sqrt{n\log n}]$. When this occurs, we have that $S_2=\Omega(n^{2k-1})$ and $|S_1-S_2|=O((n\log n)^{3/2}n^{2k-4})$, so $\frac{\mu(\{10^a10^b10^c\})}{\nu(\{10^a10^b10^c\})}\in[1-(c\log n/n)^{3/2},1+(c\log n/n)^{3/2}]$ for some constant $c$. We thus have that \[d_H^2(\mu,\nu)=\sum_{a,b,c\ge0}\left(\sqrt{\mu(\{10^a10^b10^c\})}-\sqrt{\nu(\{10^a10^b10^c\})}\right)^2\]\[\le 2n^{-100}+\sum_{a,b,c\in[np-C\sqrt{n\log n},np+C\sqrt{n\log n}]}\nu(\{10^a10^b10^c\})\left(1-\sqrt{\frac{\mu(\{10^a10^b10^c\})}{\nu(\{10^a10^b10^c\})}}\right)^2=O((\log n/n)^3).\]

It follows by Lemma \ref{HLLemma} that it requires $\Omega(n^3/\log^3n)$ samples to distinguish between traces from $x$ and $y$, as desired.
\end{proof}

\section{Conclusion and Future Work} \label{Conclusion}

We note that our work leaves several open problems, including the following:
\begin{enumerate}
    \item As noted in the introduction, Chase \cite{Chase20} very recently improved the worst-case linear trace reconstruction bound to $\exp\left(\tilde{O}(n^{1/5})\right)$. Is it possible to get a matching circular trace reconstruction bound, even just for certain lengths of strings?
    \item For worst-case strings, can one get an upper bound of $\exp\left(\tilde{O}(n^{1/3})\right)$ or even a subexponential bound for $n$ with an arbitrary prime factorization?
    \item Can one get a subpolynomial (i.e., $n^{o(1)}$) upper bound for the average case?
    \item Can one improve our current lower bound for worst-case strings, perhaps even to $n^{\omega(1)}$?
    \item All of the work we have done in this paper has primarily focused on traces with constant deletion probability $q$. However, if $q = o(1)$, the implied results are no better than the bounds we get for fixed $0 < q < 1.$ Can better bounds be obtained for circular trace reconstruction with small deletion probability (for instance, $q = 1/(\log n)^2$, or even $q = n^{-2/3}$)? In the linear case, there exist much better trace reconstruction algorithms in the low deletion probability regime (e.g., \cite{BatuKKM04, KannanM05}), so perhaps these results can be extended to circular strings. 
\end{enumerate}

For answering open problem 2, one method we attempted for getting a $\exp\left(\tilde{O}(n^{1/2})\right)$ upper bound was to look at the polynomial $P(y) P(z) P(y^{-1} z^{-1})$ for cyclotomic $n$th roots of unity $y, z$. One can establish a ``Bivariate Littlewood''-type result and the same argument as ours to show the following. Suppose that for any $a, b \in \{0, 1\}^n$, $P(y; a) P(z; a) P(y^{-1} z^{-1}; a) = P(y; b) P(z; b) P(y^{-1} z^{-1}; b)$ for all cyclotomic $n$th roots of unity $y, z$ implies that $a, b$ are equivalent up to cyclic rotation. Then, one can solve circular trace reconstruction using $\exp\left(\tilde{O}(n^{1/2})\right)$ traces. This result may in fact look obvious, as if $P(\omega_n; a) = \alpha \cdot P(\omega_n; b),$ one should expect via a simple induction argument that $P(\omega_n^k; a) = \alpha^k \cdot P(\omega_n^k; b)$ which implies that $\alpha = e^{2 \pi i r/n}$ for some $r$ (by looking at $k = n$). Thus, by rotating $a$ by $r$ elements, we will get that $P(\omega_n^k; a) = P(\omega_n^k; b)$ for all $k$, which implies $a = b$. Unfortunately, one can have cases where $P(z; a) = P(z; b) = 0$ for several choices of $z = e^{2 \pi i k/n}$, which can cause this induction argument to fail. Indeed, we believe that if such a result were true, proving it would again be a challenging number theoretic task.

We already noted some reasons in the introduction for why modifying the results of \cite{PeresZ17, HoldenPP18} to answer open problem 3 is difficult. The main reason was that one cannot efficiently find the ``start'' of the string. 

Finally, we note that a potential way of answering open problem 4 is via strings based on \cite{HoldenL18, Chase19}. One cannot use their strings directly, as their strings are equivalent up to a cyclic rotation, but perhaps an appropriate modification may improve upon our $\tilde{\Omega}(n^3)$ lower bound.

\section*{Acknowledgments}

%This research was partially supported by the MIT Akamai Fellowship, the NSF Graduate Fellowship, and by NSF-DMS grant 1949884 and NSA grant H98230-20-1-0009. 
The first author thanks Professor Piotr Indyk for many helpful discussions and feedback, Mehtaab Sawhney for pointers to some references, and Professor Bjorn Poonen for a helpful discussion on sums of roots of unity. The second author thanks Professor Joe Gallian for running the Duluth REU at which part of this research was conducted, as well as program advisors Amanda Burcroff, Colin Defant, and Yelena Mandelshtam for providing a supportive environment. The authors also thank Amanda Burcroff for helpful edits on the paper's writeup.

\newcommand{\etalchar}[1]{$^{#1}$}

\appendix

\section{Omitted Proofs} \label{Omitted}

\subsection{Proof of Proposition \ref{CircularHarderThanLinear}}

Here, we prove Proposition \ref{CircularHarderThanLinear}, which shows that circular trace reconstruction is at least as hard as linear trace reconstruction in both the worst-case and average case models for any choice of $q$.

\begin{proof}[Proof of Proposition \ref{CircularHarderThanLinear}]
    Let $m \ge 2n,$ and suppose that using $T = T(m, q)$ traces, we can solve worst-case circular trace reconstruction over length $m$ strings with failure probability $\delta.$ Then, suppose we are given $T$ traces of some unknown linear string $x$ of length $n$. We will reconstruct $x$ as follows. First, the algorithm creates a random binary string $y$ of length $m-n$. Then, the algorithm lets $x'$ be the circular string $x \circ y,$ i.e., $x$ concatenated with $y$, which has length $m$. While we do not know $x'$, given a random trace $\tilde{x}_i$ of $x$, we can create a random trace $\tilde{x}_i'$ of $x'$ by creating a random trace of $y$ (with deletion probability $q$) and appending it to $\tilde{x}_i$, and then randomly rotating it. Doing this for each trace gives us $T$ random traces of the circular string $x'$, which allows us to reconstruct $x'$ with probability $1-\delta.$ Now, the string $y$ appears exactly once (consecutively) in the circular string $x'$ with failure probability exponentially small in $n$ since $m \ge 2n$, and since we know $y$, we would be able to find the unique copy of $y$ in $x'$ and thus recover the linear string $x$ with failure probability $\delta + e^{-\Omega(n)}$.
    
    The same argument works in the average case. Suppose using $T = T(m, q)$ traces, we can solve average-case circular trace reconstruction with probability $\delta,$ where the average string is generated by creating a uniformly random binary (linear) string and making it circular. Then, if given $T$ random traces of a random linear string $x$ of length $n$, our algorithm works the same way: creating a random string $y$ of length $m-n,$ appending it to $x$, reconstructing the circular string $x' = x \circ y,$ and then recovering $x$ since with $1 - e^{-\Omega(n)}$ probability, there is a unique copy of $y$ in $x'.$
\end{proof}

\subsection{Proof of Lemma \ref{CreatingG}}

Here, we prove Lemma \ref{CreatingG}, which gives us the unbiased estimator of $\prod_{i = 1}^{m} P(z_i; x)$. To do so, we first note a simple proposition about complex numbers.

\begin{proposition}  (folklore, see, e.g., \cite{Narayanan20}) \label{Complex}
    Let $z$ be a complex number with $|z| = 1$ and $|\arg z| \le \theta.$ Then, for any $0 < p < 1,$ $\left|\frac{z-(1-p)}{p}\right| \le 1 + \frac{\theta^2}{p^2}.$
\end{proposition}

\begin{proof}[Proof of Lemma \ref{CreatingG}]
For some $1 \le k \le m,$ fix some complex numbers $w_1, \dots, w_k$ and consider the random variable
\[f(\tilde{x}, w) := \sum\limits_{1 \le i_1 < i_2 < \dots < i_k \le n} \tilde{x}_{i_1} \cdots \tilde{x}_{i_k} w_1^{i_1} w_2^{i_2-i_1} \cdots w_k^{i_k-i_{k-1}}\]
    for $w = (w_1, \dots, w_k)$, which is a random variable since $\tilde{x}$ is random. 
%    Given $(w_1, \dots, w_k)$ and $\tilde{x},$ since there are at most ${n \choose k}$ terms, this can be computed in time $n^{O(k)}.$
    
    We first describe $\BE[f(\tilde{x}, w)]$ and choose appropriate values for $w_1, \dots, w_k$. First, we can rewrite
\[f(\tilde{x}, w) = \sum\limits_{\substack{i_1, \dots, i_k \ge 1\\ i_1 + \dots + i_k \le n}} \tilde{x}_{i_1} \tilde{x}_{i_1+i_2} \cdots \tilde{x}_{i_1+i_2+\dots+i_k} w_1^{i_1} w_2^{i_2} \cdots w_k^{i_k}.\]
    For any $j_1, \dots, j_k$, note that $\tilde{x}_{i_1}$ coming from $x_{j_1}$, $\tilde{x}_{i_1+i_2}$ coming from $x_{j_1+j_2}$, etc. means that $j_1 \ge i_1, j_2 \ge i_2, \dots, j_k \ge i_k$. Moreover, even in this case, this will only happen with probability
\[\prod\limits_{r = 1}^{k} \left(p \cdot {j_r-1 \choose i_r-1} p^{i_r-1} q^{j_r-i_r} \right) = p^{\sum i_r} q^{\sum (j_r-i_r)} \prod\limits_{r = 1}^{k} {j_r - 1 \choose i_r - 1}.\]
    Therefore, we have that
\begin{align*}
\BE[f(\tilde{x}, w)] &= \sum\limits_{\substack{i_1, \dots, i_k \ge 1 \\ j_r \ge i_r \\ j_1 + \dots + j_k \le n}} \prod\limits_{r = 1}^{k} \left({j_r - 1 \choose i_r - 1} p^{i_r} q^{j_r-i_r} x_{j_1+\dots+j_r} w_r^{i_r}\right)\\
&=\sum\limits_{\substack{j_1, \dots, j_k \ge 1 \\ j_1 + \dots + j_k \le n}} \prod\limits_{r = 1}^{k} \left(p w_r x_{j_1+\dots+j_r} \cdot \sum\limits_{i_r = 1}^{j_r} {j_r - 1 \choose i_r - 1} p^{i_r-1} q^{j_r-i_r} w_r^{i_r-1}\right) \\
&=\sum\limits_{\substack{j_1, \dots, j_k \ge 1 \\ j_1 + \dots + j_k \le n}} \prod\limits_{r = 1}^{k} \left(p w_r x_{j_1+\dots+j_r} \cdot (p w_r + q)^{j_r-1} \right) \\
&= p^k \frac{w_1 \cdots w_k}{(pw_1+q) \cdots (pw_k+q)} \cdot \sum\limits_{\substack{j_1, \dots, j_k \ge 1 \\ j_1 + \dots + j_k \le n}} x_{j_1} \cdots x_{j_1+\dots+j_k} (pw_1+q)^{j_1} \cdots (pw_k+q)^{j_k}.
\end{align*}

    Now, fix $k \le m$ and fix a sequence $B = (B_1, \dots, B_k)$ of strictly nested nonempty subsets of $[m]$ with $B_1 = [m]$. By this, we mean that $[m] = B_1 \supsetneq B_2 \supsetneq \cdots \supsetneq B_k \neq \emptyset$. Now, for $1 \le r \le k$, define $w_{B, r} := \frac{1}{p} \left(\left(\prod_{i \in B_r} z_i\right)-q\right)$, $w_B := (w_{B, 1}, \dots, w_{B, k}),$ and $C_r := B_r \backslash B_{r+1}$ for $1 \le r \le k-1$ and $C_k := B_k.$ Finally, for any set $S \subset [m]$, define $z_S := \prod_{i \in S} z_i$. Then,
\begin{align*}
\BE\left[f(\tilde{x}, w_B)\right] &= p^k \cdot \frac{w_{B, 1} \cdots w_{B, k}}{z_{B_1} z_{B_2} \cdots z_{B_k}} \cdot \sum\limits_{\substack{j_1, \dots, j_k \ge 1 \\ j_1 + \dots + j_k \le n}} x_{j_1} \cdots x_{j_1+\dots+j_k} z_{B_1}^{j_1} \cdots z_{B_k}^{j_k} \\
&= p^k \cdot \frac{w_{B, 1} \cdots w_{B, k}}{z_{B_1} z_{B_2} \cdots z_{B_k}} \cdot \sum\limits_{1 \le i_1 < i_2 < \dots < i_k \le n} x_{i_1} \cdots x_{i_k} z_{C_1}^{i_1} \cdots z_{C_k}^{i_k},
\end{align*}
    where we have written $i_r = j_1 + j_2 + \dots + j_r$ for all $1 \le r \le k$. This implies that
\begin{align*}
\prod_{k = 1}^{m}\left(\sum\limits_{i = 1}^{n} x_i z_k^i\right) &= \sum\limits_{\substack{1 \le k \le m \\ [m] = B_1 \supsetneq \cdots \supsetneq B_k \neq \emptyset}} \sum\limits_{1 \le i_1 < i_2 < \dots < i_k \le n} x_{i_1} \cdots x_{i_k} z_{C_1}^{i_1} \cdots z_{C_k}^{i_k} \\
&= \sum\limits_{\substack{1 \le k \le m \\B = (B_1, \dots, B_k)}} p^{-k} \cdot \frac{z_{B_1} \cdots z_{B_k}}{w_{B, 1} \cdots w_{B, k}} \cdot \BE\left[f(\tilde{x}, w_B)\right].
\end{align*}
    The first line is true by expanding and using the fact that $x_i = x_i^{b_i}$ for all $b_i \in \BN$, as $x \in \{0, 1\}$.
    
    Now, let 
\[g_m(\tilde{x}, Z) := \sum\limits_{\substack{1 \le k \le m \\B = (B_1, \dots, B_k)}} p^{-k} \cdot \frac{z_{B_1} \cdots z_{B_k}}{w_{B, 1} \cdots w_{B, k}} \cdot f(\tilde{x}, w_B).\]
    %Since the number of tuples $(b_1, b_2, \dots, b_k)$ that add to $m$ is $2^{O(m)}$ and since $k \le m$ for all tuples $B$, $g_m(\tilde{x}, z)$ can be computed in $n^{O(m)}$ time. 
    For fixed $p, q, n,$ note that $g_m(\tilde{x}, Z)$ is indeed only a function of $\tilde{x},$ $Z$, and $m$, as the $w_{B, r}$'s are determined given $Z$. Importantly, there is no dependence of $g$ on $x$. Then,
\[\BE[g_m(\tilde{x}, Z)] = \prod_{k = 1}^{m}\left(\sum\limits_{i = 1}^{n} x_i z_k^i\right).\]

    Finally, we provide bounds on $f(\tilde{x}, w_B)$ that will give us our bounds on $g_m(\tilde{x}, Z).$ If $|z_i| = 1$ and $|\arg z_i| \le \frac{1}{L}$ for all $i$, then for $B = (B_1, \dots, B_r),$ $z_{B_r} = \prod_{i \in B_r} z_i$ has magnitude $1$ and argument at most $\frac{m}{L}$ in absolute value. Therefore, $|w_{B, r}| = \frac{1}{p}\left(\left(\prod_{i \in B_r} z_i\right)-q\right) \le 1 + O\left(\frac{m^2}{p^2L^2}\right) = \exp\left(O\left(\frac{m^2}{p^2L^2}\right)\right)$, by Proposition \ref{Complex}. This means that for any $i_1 + \dots + i_k \le n,$
\[|w_{B, 1}^{i_1} w_{B, 2}^{i_2} \cdots w_{B, k}^{i_k}| \le \exp\left(O\left(\frac{m^2}{p^2L^2} \cdot n\right)\right).\]
    Since $f(\tilde{x}, w_B)$ is the sum of $w_{B, 1}^{i_1} w_{B, 2}^{i_2} \cdots w_{B, k}^{i_k}$ over all $i_1, \dots, i_k \ge 1$ with $i_1 + \dots + i_k \le n,$ there are at most $n^k$ choices of $i_1, \dots, i_k,$ so we have that
\[|f(\tilde{x}, w_B)| \le \exp\left(O\left(\frac{m^2}{p^2L^2} \cdot n\right)\right) \cdot n^k.\]
    
    Therefore,
\[|g_m(\tilde{x}, Z)| \le \sum_{\substack{1 \le k \le m \\ B = (B_1, \dots, B_k)}} p^{-k} \cdot \left|\frac{z_{B_1} \cdots z_{B_k}}{w_{B, 1} \cdots w_{B, k}}\right| \cdot |f(\tilde{x}, w_B)| \le \sum_{\substack{1 \le k \le m \\ B = (B_1, \dots, B_k)}} p^{-k} \cdot \exp\left(O\left(\frac{m^2 n}{p^2 L^2}\right)\right) \cdot n^k\]
\[ \le (p^{-1} nm)^{O(m)} \cdot \exp\left(O\left(\frac{m^2 n}{p^2 L^2}\right)\right).\]
    The final equation follows from the observation that the number of sequences $(B_1, \dots, B_k)$ is at most $m^m,$ since each $i \in [m]$ has some final subset $j$ such that $i \in B_j$ but $i \not\in B_{j+1}$.
\end{proof}
    
\subsection{Proof of Theorem \ref{NT:Main}} \label{NT:Subsection}

We let $\omega_k$ denote $e^{2 \pi i/k}$ for $k \ge 1.$ When dealing with a string of length $n$, we write $\omega := \omega_n.$

    In the case where $n$ is prime, we already proved it in Proposition \ref{NT:Prime}.

    Next, we prove it in the case that $n = p \cdot q$ for $p, q$ odd primes. We will first need a simple lemma, which is likely folklore, though we give a proof regardless.
    
\begin{lemma} \label{NT:Lemma}
    Let $p, q$ be distinct primes, and suppose that $(1-\omega_p)|(\sum_{i = 0}^{q-1} b_i \omega_q^i)$ in $\BQ[\omega_{pq}].$ Then, we in fact have that $p|(\sum_{i = 0}^{q-1} b_i \omega_q^i)$ in $\BQ[\omega_{pq}]$.
\end{lemma}

\begin{proof}
    Note that $(1-\omega_p)|(\sum_{i = 0}^{q-1} b_i \omega_q^i)$ implies that $(1-\omega_p)^p|(\sum_{i = 0}^{q-1} b_i \omega_q^i)^p.$ But $p|(1-\omega_p)^p,$ so $p|(\sum_{i = 0}^{q-1} b_i \omega_q^i)^p.$ Now, using Frobenius Endomorphism, we have that $(\sum_{i = 0}^{q-1} b_i \omega_q^i)^p \equiv \sum_{i = 0}^{q-1} b_i \omega_q^{i \cdot p} \Mod p,$ so $p|(\sum_{i = 0}^{q-1} b_i \omega_q^{i \cdot p}).$ But since $p \neq q,$ we have that $\omega_q^p$ and $\omega_q$ are Galois conjugates, so we therefore have that $p|(\sum_{i = 0}^{q-1} b_i \omega_q^i),$ as desired.
\end{proof}

\begin{lemma}
    Theorem \ref{NT:Main} is true when $n = p \cdot q,$ where $p, q$ are distinct odd primes.
\end{lemma}

\begin{proof}
    First, we have that $\sum_{i = 0}^{n-1} a_i = \omega^{c_0} \sum_{i = 0}^{n-1} b_i$ and since $a_1, \dots, a_n,b_1, \dots, b_n \in \{0, 1\}$, this means that $\omega^{c_0}$ is real and thus equals $1$. So, $\sum_{i = 0}^{n-1} a_i = \sum_{i = 0}^{n-1} b_i$. Next, we have that $\sum_{i = 0}^{n-1} a_i \omega_p^i = \omega^{c_q} \cdot \sum_{i = 0}^{n-1} b_i \omega_p^i,$ since $\omega_p = \omega^q.$ Therefore, since the $a_i$'s are all integers, this implies that either $\sum_{i = 0}^{n-1} a_i \omega_p^i = \sum_{i = 0}^{n-1} b_i \omega_p^i = 0$ or $\omega^{c_q} = \frac{\sum a_i \omega_p^i}{\sum b_i \omega_p^i} \in \BQ[\omega_p]$. Thus, by Theorem \ref{RootsOfUnity}, $\omega^{c_q}$ actually equals $\omega_p^k$ for some $k$. Likewise, $\omega^{c_p}$ actually equals $\omega_q^\ell$ for some $\ell$, so we have that
\begin{equation*}
    \label{NT:a} \sum_{i = 0}^{n-1} a_i \omega_p^i = \omega_p^k \sum_{i = 0}^{n-1} b_i \omega_p^i, \hspace{0.5cm} \sum_{i = 0}^{n-1} a_i \omega_q^i = \omega_q^{\ell} \sum_{i = 0}^{n-1} b_i \omega_q^i.
\end{equation*}
    
    Therefore, by the Chinese Remainder theorem, we can cyclically shift $\{b_i\}$ by something that is $k$ modulo $p$ and $\ell$ modulo $q$ to get some sequence $\{b'_i\}$ so that
\begin{equation*}
    \label{NT:b}\sum_{i = 0}^{n-1} a_i = \sum_{i = 0}^{n-1} b_i', \hspace{0.5cm} \sum_{i = 0}^{n-1} a_i \omega_p^i = \sum_{i = 0}^{n-1} b'_i \omega_p^i, \hspace{0.5cm} \text{and} \hspace{0.5cm} \sum_{i = 0}^{n-1} a_i \omega_q^i = \sum_{i = 0}^{n-1} b'_i \omega_q^i.
\end{equation*}

    Without loss of generality, we can therefore pretend that $k = \ell = 0,$ so in fact we have $b_i' = b_i$ for all $i$. Now, suppose that $\sum a_i \omega^i = \omega^m \cdot \sum b_i \omega^i.$ Our goal is to show that $p|m$ and $q|m$, so that $\omega^m = 1.$ Assume the contrary, WLOG that $q \nmid m.$ Then, we can write
\begin{equation}
    \label{NT:c} \sum_{i = 0}^{n-1} (a_i-b_i) \omega^i = (\omega^m-1) \cdot \sum_{i = 0}^{n-1} b_i \omega^i. 
\end{equation}
    Now, choose integers $r, s$ so that $r \cdot q + 1 = s \cdot p.$ Then, we have that $\omega_q^{i \cdot s} -\omega^i = \omega^{i \cdot s \cdot p} - \omega^i = \omega^i\left(\omega^{i \cdot r \cdot q}-1\right) = \omega^i\left(\omega_p^{i \cdot r} - 1\right)$, which is a multiple of $\omega_p-1$. Therefore, we have that $1-\omega_p$ divides
\[\sum_{i = 0}^{n-1} (a_i-b_i) \cdot \left(\omega^i - \omega_q^{i \cdot s}\right) = \sum_{i = 0}^{n-1} (a_i-b_i) \omega^i - \sum_{i = 0}^{n-1} (a_i-b_i) \omega_q^{i \cdot s} = \sum_{i = 0}^{n-1} (a_i-b_i) \omega^i.\]
    The last equality in the above line follows since $\sum_{i = 0}^{n-1} a_i \omega_q^i = \sum_{i = 0}^{n-1} b_i \omega_q^i,$ and since $s$ is relatively prime to $q$, this means $\omega_q^s$ is a Galois conjugate of $\omega_q,$ so $\sum_{i = 0}^{n-1} a_i \omega_q^{i \cdot s} = \sum_{i = 0}^{n-1} b_i \omega_q^{i \cdot s}.$

    Now, since $q \nmid m,$ we have that either $\omega^m-1$ is a unit in $\BZ[\omega]$ (if $p \nmid m$) or $\omega^m-1|q$ (if $p|m$). Therefore, by Equation \eqref{NT:c}, we have that
\begin{equation*}
    (1-\omega_p) \biggr\vert q \cdot \sum_{i = 0}^{n-1} b_i \omega^i \Rightarrow (1-\omega_p) \biggr\vert \sum_{i = 0}^{n-1} b_i \omega^i,
\end{equation*}
    since $(1-\omega_p),$ $(q)$ are relatively prime as ideals. Now, recalling that $\omega^i \equiv \omega_q^{i \cdot s} \Mod {1-\omega_p},$ we have that $(1-\omega_p)|\sum_{i = 0}^{n-1} b_i \omega_q^{i \cdot s}.$
    
    By Lemma \ref{NT:Lemma}, we have that $p|\sum_{i = 0}^{n-1} b_i \omega_q^{i \cdot s}$. Since $\omega_q^s$ and $\omega_q$ are Galois conjugates, this also means that $p|\sum_{i = 0}^{n-1} b_i \omega_q^i.$ Now, for $0 \le j \le q-1,$ let $d_j = b_j+b_{j+q} + \dots + b_{j+(p-1)q}$. We have that $p|\sum_{j = 0}^{q-1} d_j \omega_q^{j},$ so $\sum_{j = 0}^{q-1} \frac{d_j}{p} \omega_q^j$ is an algebraic integer in $\BQ[\omega_q].$ Therefore, $d_0 \equiv d_1 \equiv \dots \equiv d_{q-1} \Mod p$. Since $0 \le d_i \le p$ for all $i$, we either have that $d_0 = d_1 = \dots = d_{q-1},$ or $d_0, d_1, \dots, d_{q-1} \in \{0, p\}.$
    
    Likewise, we also have that $\sum b_i \omega^i = \omega^{-m} \cdot \sum a_i \omega^i,$ where $q \nmid (-m).$ Therefore, if $c_j = a_j + a_{j+q} + \cdots + a_{j + (p-1)q}$ for each $0 \le j \le q-1,$ we either have that $c_0 = c_1 = \cdots = c_{q-1},$ or $c_0, c_1, \dots, c_{q-1} \in \{0,p\}.$
    
    Now, suppose that $d_0, d_1, \dots, d_{q-1} \in \{0, p\}.$ This means that for all $0 \le j \le q-1,$ $b_j = b_{j+q} = \cdots = b_{j + (p-1) q},$ so $b_j \omega^j + b_{j+q} \omega^{j+q} + \cdots + b_{j + (p-1) q} \omega^{j + (p-1) q} = 0$ for all $0 \le j \le p-1.$ Importantly, this means $\sum b_j \omega^j = 0$. But since $\sum a_i \omega^i = \omega^m \cdot \sum b_i \omega^i$ for some $m \in \BZ,$ this also means that $\sum a_i \omega^i = 0,$ so in fact we do have that $\sum b_i \omega^i = \sum a_i \omega^i.$ Likewise, if $a_0, a_1, \dots, a_{q-1} \in \{0, p\}$ we also have that $\sum b_i \omega^i = \sum a_i \omega^i = 0$ by a symmetric argument.
    
    Otherwise, we are dealing with the case where $d_0 = d_1 = \cdots = d_{q-1}$ and $c_0 = c_1 = \cdots = c_{q-1}.$ But then, $0 = \sum_{j = 0}^{q-1} d_j \omega_q^j = \sum_{i = 0}^{n-1} b_i \omega_q^i$ and $0 =  \sum_{j = 0}^{q-1} c_j \omega_q^j = \sum_{i = 0}^{n-1} a_i \omega_q^i.$ Recall that $\sum a_i \omega^i = \omega^m \cdot \sum b_i \omega^i$ and that we assumed $q \nmid m.$ If $p|m,$ then if $m = p \cdot t$, we have $\sum a_i \omega^i = \sum b_{i- p \cdot t} \omega^i,$ where $i - p \cdot t$ is done modulo $n$. Moreover, we have that $\sum b_{i-p \cdot t} \omega_p^i = \sum b_i \omega_p^{i + p \cdot t} = \sum b_i \omega_p^i$, $\sum b_{i-o \cdot t} \omega_q^i = \sum b_i \omega_q^{i + p \cdot t} = \omega_q^{p \cdot t} \cdot \sum b_i \omega_q^i = 0,$ and $\sum b_{i - p \cdot t} = \sum b_i.$ Therefore, by shifting $b$ by $p \cdot t,$ we have that $\sum a_i \omega^{k \cdot i} = \sum b_i \omega^{k \cdot i}$ for $k = 0, 1, p,$ and $q$, and therefore for all $0 \le k \le n-1.$
    
    The other case is that $p \nmid m.$ In this case, we can define $e_j = a_j + a_{j+p} + \cdots + a_{j+(q-1)p}$ and $f_j = b_j + b_{j+p} + \cdots + b_{j+(q-1)p}$ for $0 \le j \le p-1.$ By the same argument as before, either $e_0 = e_1 = \cdots = e_{p-1}$ or $e_0, e_1, \dots, e_{p-1} \in \{0, q\},$ and either $f_0 = f_1 = \cdots = f_{p-1}$ or $f_0, f_1, \cdots, f_{p-1} \in \{0, p\}.$ Again, either $e_0, e_1, \dots, e_{p-1} \in \{0, q\}$ or $f_0, f_1, \dots, f_{q-1} \in \{0, q\}$ implies that $\sum a_i \omega^i = \sum b_i \omega^i = 0.$ Therefore, the final case to deal with is if $c_0 = c_1 = \cdots = c_{q-1},$ $d_0, d_1 = \cdots = d_{q-1},$ $e_0 = e_1 = \cdots = e_{q-1},$ and $f_0 = f_1 = \cdots = f_{q-1}.$ As we have seen, the first two equations imply that $0 = \sum_{i = 0}^{n-1} a_i \omega_q^i = \sum_{i = 0}^{n-1} b_i \omega_q^i.$ Thus, the same argument applied to the last two equations implies that $0 = \sum_{i = 0}^{n-1} a_i \omega_p^i = \sum_{i = 0}^{n-1} b_i \omega_p^i.$ As a result, we can shift the sequence $b$ by $m$, since $\sum a_i \omega^i = \sum b_{i-m} \omega^i,$ but we will still have that $\sum a_i = \sum b_{i-m},$ $\sum a_i \omega_p^i = \sum b_i \omega_p^i = \sum b_{i-m} \omega_p^i = 0,$ and $\sum a_i \omega_q^i = \sum b_i \omega_q^i = \sum b_{i-m} \omega_q^i = 0.$
\end{proof}

We now show that Theorem \ref{NT:Main} is true when $n$ is the square of an odd prime.

\begin{proposition}
    Theorem \ref{NT:Main} is true if $n=p^2$ is the square of an odd prime.
\end{proposition}

\begin{proof}
By shifting, we may without loss of generality assume that $\sum a_i\omega^i=\sum b_i\omega^i$, so $P(x)=\sum(a_i-b_i)x^i$ has $\omega$ as a root. Thus, $1+x^p+\cdots+x^{n-p}\mid P(x)$, which means that $a_i-b_i=a_{i+p}-b_{i+p}=\cdots=a_{i+n-p}-b_{i+n-p}$, where indices are taken mod $n$. Thus, if it is not the case that $a_i=a_{i+p}=\cdots=a_{i+n-p}$, equivalently that $b_i=b_{i+p}=\cdots=b_{i+n-p}$, then we must have that $a_i=b_i,a_{i+p}=b_{i+p},\ldots,a_{i+n-p}=b_{i+n-p}$ since $a_j,b_j\in\{0,1\}$.

Let $z=\omega^p$. We have that \[\sum a_iz^i=(a_0+a_p+\cdots+a_{n-p})+(a_1+a_{p+1}+\cdots+a_{n-p+1})z+\cdots+(a_{p-1}+a_{2p-1}+\cdots+a_{n-1})z^{p-1},\]\[\sum b_iz^i=(b_0+b_p+\cdots+b_{n-p})+(b_1+b_{p+1}+\cdots+b_{n-p+1})z+\cdots+(b_{p-1}+b_{2p-1}+\cdots+b_{n-1})z^{p-1}.\]

Thus, $\frac{\sum a_iz^i}{\sum b_iz^i}\in\mathbb Q[z]$, so $p\mid c_p$ and we have that $\sum a_iz^i=z^m\sum b_iz^i$ for some $m\in\mathbb Z$. It follows that $\{a_i+a_{i+p}+\cdots+a_{i+n-p}\}$ and $\{b_i+b_{i+p}+\cdots+b_{i+n-p}\}$ are cyclic shifts of each other. Since these sequences are of length $p$, this means that they are equal. We already know that $a_i+a_{i+p}+\cdots+a_{i+n-p}\notin\{0,p\}\implies a_{i+p\ell}=b_{i+p\ell}$ for all $\ell$. But it is also the case that $a_i+a_{i+p}+\cdots+a_{i+n-p}=b_i+b_{i+p}+\cdots+b_{i+n-p}\in\{0,p\}\implies a_{i+p\ell}=b_{i+p\ell}$ for all $\ell$ since $a_j,b_j\in\{0,1\}$. Thus, we have shown that $a_i=b_i$ for all $i$, so we are done.
\end{proof}

\begin{proposition}
    Theorem \ref{NT:Main} is true if $n = 2p$, i.e., $n$ is twice a prime.
\end{proposition}

\begin{proof}
    If $p = 2,$ i.e., $n = 4,$ then if $\sum a_i = \sum b_i$ but the sequences $\{a_i\}$ and $\{b_i\}$ are not equal up to a cyclic rotation, then up to cyclic rotations, we either have $\{a_i\} = \{1, 0, 1, 0\}$ and $\{b_i\} = \{1, 1, 0, 0\}$ or vice versa. But then, $\sum a_i (-1)^i = \pm 2$ and $\sum b_i (-1)^i = 0.$
    
    If $p$ is an odd prime, then note that the minimal polynomial of $\omega = \omega_n$ is $1 + x^2 + \dots + x^{2(p-1)}.$ Now, suppose that $a, b$ are rotated so that if $P(x) := \sum_{i = 0}^{n-1} a_i x^i$ and $Q(x) := \sum_{i = 0}^{n-1} b_i x^i$, then $P(\omega) = Q(\omega).$ Therefore, $(1 + x^2 + \dots + x^{2(p-1)})|\sum_{i = 0}^{2p-1} (a_i-b_i)x^i.$ Since $a_i-b_i \in \{-1, 0, 1\}$ for all $i$, we must have that $\sum_{i = 0}^{2p-1} (a_i-b_i)x^i = (1 + x^2 + \dots + x^{2(p-1)}) \cdot R(x),$ where $R(x)$ must be either $0, \pm 1$, and $\pm x \pm 1.$ However, since $\sum a_i = \sum b_i,$ we have that $P(1)-Q(1) = 0 = (p-1) \cdot R(1),$ so $R(1) = 0.$ Thus, $R(x)$ must equal either $0$, $x-1$, or $1-x.$
    
    If $R(x) = x-1,$ then $a_i-b_i = 1$ for all odd $i$ and $-1$ for all even $i$, which means that $a_i = 1$ if and only if $i$ is odd, but $b_i = 0$ if and only if $i$ is even. Since $n$ is even, this means that $\{a_i\}$ and $\{b_i\}$ are the same sequence, up to a rotation by $1$. The same is true if $R(x) = 1-x$ by symmetry between $a$ and $b$. Finally, if $R(x) = 0,$ then $a_i-b_i = 0$ for all $i$, so $a_i = b_i$ for all $i$, and thus the sequences $\{a_i\}$ and $\{b_i\}$ are the same.
\end{proof}

Finally, we remark that the statement is false for numbers with $3$ or more prime factors, which concludes the proof of Theorem \ref{NT:Main}. Suppose that $n=abc$ with $a,b,c>1$. Let $A=\{1,a+1,\ldots,ab-a+1,a,ab+a,\ldots,abc-ab+a\}$ and $B=\{1,a+1,\ldots,ab-a+1,0,ab,\ldots,abc-ab\}$. Consider circular strings $a$ and $b$ of length $n$ with $1$s in positions given by $A$ and $B$, respectively. Let $P(x)=\sum_{i\in A}x^i$ and $Q(x)=\sum_{i\in B}x^i$. We have that $P(x)-Q(x)=(x^a-1)\cdot\frac{x^{abc}-1}{x^{ab}-1}$ and $P(x)-x^aQ(x)=x(1-x^{ab})$. Thus, for all $k$, $\frac{P(\omega^k)}{Q(\omega^k)}$ is a power of $\omega$, so the conditions of Theorem \ref{NT:Main} hold. However, $a$ and $b$ are not cyclic shifts of each other.

\end{document}